\documentclass[11pt, a4paper]{amsart}

\usepackage[utf8]{inputenc}
\usepackage[T1]{fontenc}        
\usepackage{lmodern}            
\usepackage{microtype}

\usepackage[margin=1in]{geometry}       

\usepackage{appendix}
\usepackage{svg}

\usepackage{amsmath,amssymb,mathtools,amsthm}
\usepackage{bm}                 
\usepackage{gensymb}            

\usepackage{aliascnt}
\usepackage{pdfpages}

\usepackage{pgfplots}
\pgfplotsset{compat=1.18}

\usepackage[numbers, square, sort&compress]{natbib}    
\usepackage{xcolor}
\usepackage{booktabs}           
\usepackage[colorlinks=true, linkcolor=blue, citecolor=blue, urlcolor=black]{hyperref}
\usepackage[capitalise, nameinlink]{cleveref}

\newtheorem{theorem}{Theorem}
\newtheorem{lemma}[theorem]{Lemma}
\crefname{lemma}{Lemma}{Lemmas}
\Crefname{lemma}{Lemma}{Lemmas}

\newtheorem{corollary}[theorem]{Corollary}
\crefname{corollary}{Corollary}{Corollaries}
\Crefname{corollary}{Corollary}{Corollaries}

\newcommand{\R}{\mathbb{R}}
\renewcommand{\P}{\mathbb{P}}
\newcommand{\E}{\mathbb{E}}
\newcommand{\N}{\mathbb{N}}
\newcommand{\1}{\mathbf{1}}
\newcommand{\visible}{\leftrightarrow{}}

\newtheorem{proposition}[theorem]{Proposition}

\theoremstyle{definition}
\newtheorem{definition}[theorem]{Definition}

\theoremstyle{remark}

\theoremstyle{definition}
\newtheorem{assumption}{Assumption}
\crefname{assumption}{Assumption}{Assumptions}
\Crefname{assumption}{Assumption}{Assumptions}

\usepackage{csquotes}
\usepackage{dsfont}
\usepackage{stmaryrd}
\usepackage{algorithm}
\usepackage[noend]{algpseudocode}
\usepackage{hyperref}

\DeclareMathOperator*{\argmax}{arg\,max}

\definecolor{darkgreen}{RGB}{0,100,0}

\algdef{SE}[DOWHILE]{Do}{doWhile}
  {\algorithmicdo}
  {\algorithmicwhile\ }

\newcommand{\invisible}[1]{}

\algnewcommand{\algorithmicgoto}{\textbf{go to} line}%
\algnewcommand{\Goto}[1]{\algorithmicgoto~\ref{#1}}%

\title{Discretization-free exact recovery in geometric community detection}

\author{Maarten Hoeneveld}
\author{Moritz Otto}
\author{Raphaël Sala}

\address{Mathematical Institute\\
Leiden University\\
Leiden, The Netherlands}
\email{maartenhoeneveld@gmail.com}

\address{Mathematical Institute\\
Leiden University\\
Leiden, The Netherlands}
\email{m.f.p.otto@math.leidenuniv.nl}

\address{Institut de Mathématiques\\
Toulouse University\\
Toulouse, France}
\email{contact.r.sala@gmail.com}

\date{\today}

\begin{document}

\begin{abstract}
Geometric community detection seeks to recover latent communities in networks where connectivity depends jointly on community structure and continuous spatial geometry. Existing exact-recovery approaches typically discretize the underlying space, which can impose restrictive structural assumptions on the connectivity functions. We develop a polynomial-time, discretization-free algorithm for exact recovery in the Geometric Hidden Community Model (GHCM), operating directly on the continuous geometry. Our method succeeds even when connectivity functions coincide on a nontrivial portion of their visibility range and when two communities can be distinguished only through their connectivity to a third community. We prove exact recovery under these weaker conditions and provide experiments showing that the algorithm succeeds beyond the theoretically guaranteed regime.
\end{abstract}

\maketitle

\section{Introduction}

Community detection is a fundamental problem in graph learning and network
science. While classical models such as the Stochastic Block Model (SBM)
capture community-dependent connectivity, many real-world networks are
embedded in continuous physical or metric spaces, where interactions depend
jointly on community membership and spatial proximity. This motivates
geometric extensions of the SBM, in which observed edge weights depend on
both the latent labels of their endpoints and their geometric distance.

We study exact community recovery in the Geometric Hidden Community Model
(GHCM) introduced by Gaudio et al.~\cite{gaudio2025incguan}, where vertices
are positioned in a continuous spatial domain and edge distributions depend
on both community labels and spatial distance. Existing efficient recovery
algorithms for the GHCM navigate this continuous geometry through spatial
discretization: the domain is partitioned into small cells and community
labels are propagated across sufficiently occupied regions. This raises two
natural questions: \emph{Is spatial discretization necessary for
computationally efficient exact recovery, and how restrictive are the
structural distinctness assumptions used in existing analyses?}

We address both questions through two complementary contributions:
\begin{itemize}
    \item \textbf{Discretization-free exact recovery.}
    We introduce a polynomial-time exact-recovery algorithm that operates
    directly on the continuous vertex locations, without partitioning the
    underlying space into cells. The algorithm combines initialization on a
    logarithmic-size region, greedy propagation based on the available
    statistical information, and a final local likelihood refinement. Its
    propagation phase can be implemented in
    $\mathcal O(n\log^2 n)$ operations.

\item \textbf{Relaxed structural conditions.}
Our propagation analysis allows a pair of communities to be
distinguished through witness communities different from the pair
itself, substantially weakening the structural assumptions required by
previous propagation arguments. We establish these guarantees for both
our discretization-free algorithm and a block-based propagation
algorithm, showing that this improvement is independent of the removal
of spatial discretization.
\end{itemize}

Under a witness-connectivity condition, both algorithms achieve exact
recovery whenever the information-theoretic recovery condition holds.
Consequently, whenever the information-theoretic condition itself implies
witness connectivity, our algorithms attain the sharp information-theoretic
threshold. This leaves an intermediate regime in which exact recovery is
information-theoretically possible but our sufficient connectivity condition
is not satisfied. Numerical experiments indicate that the
discretization-free algorithm continues to succeed well inside this regime,
suggesting that the additional connectivity condition may reflect a
limitation of the current analysis or of block-based propagation rather than
an intrinsic barrier to recovery.

Our implementation and numerical experiments are publicly available at
\url{https://github.com/Sala2Code/GHCM-Greedy-Exact-Recovery}.

\textbf{Related Work.} The study of community detection in spatially embedded networks was initiated by Baccelli and Sankararaman \cite{Sankararaman2018} via the Geometric Stochastic Block Model (GSBM). Gaudio et al. \cite{gaudio2025dec} subsequently developed an efficient exact-recovery procedure for the GSBM based on local propagation and likelihood refinement, complementing the lower bounds of Abbe et al. \cite{Abbe+2020}. Avrachenkov et al. \cite{avrachenkov2024communitydetectionblockmodels} established exact recovery for a one-dimensional geometric block model with general distance-dependent edge probabilities. This framework was generalized by Gaudio et al. \cite{gaudio2025incguan} to the GHCM, accommodating arbitrary joint edge-weight distributions. More recently, Gaudio and Jin \cite{gaudio2026jan} characterized the information-theoretic threshold for exact recovery in the distance-dependent GHCM. Related models have also been explored in other geometric settings, such as the Soft Geometric Block Model with unknown locations \cite{avrachenkov2021higher} and spherical geometric block models \cite{galhotra2018geometric,galhotra2023community}.

\section{Model and Main Result}
\label{sec:model}

We consider the following  GHCM from \cite{gaudio2026jan}, which allows the edge weight distribution between two vertices to depend on their distance, in addition to their community labels.

 \begin{definition}
Let $n \in \N$ be a scaling parameter. Let $\lambda > 0$ be a constant, $d \in \N$ be the dimension, $k \in \N$ be the number of communities, and let $(\mathbb X, \mathcal X)$ be a measurable space of edge weights with $\sigma$-finite measure $\mu$. Let $Z$  with $|Z| = k$ be the set of community labels. Let $\pi \in \R^k$ represent the prior probabilities of each community. For any $i,j\in Z$ and $y\in \mathbb R$, let $P_{ij}(y)=P_{ji}(y)$ be the symmetric edge-weight distribution between two vertices in communities $i$ and $j$ at distance $y$. Assume that $P_{ij}(y) \ll \mu$ for all $y \in \R$ and let $p_{ij}(\cdot ; y)$ denote the probability density of $P_{ij}(y)$ with respect to $\mu$. Then, we sample $G \sim \text{GHCM}(\lambda, n, r, \pi, P(\cdot), d)$ as follows:
\begin{enumerate}
    \item The locations of the vertices are generated by a Poisson point process with intensity $\lambda$ over $\mathcal S_{d, n} := [-n^{1/d}/2, n^{1/d}/2]^d$, the cube of volume $n$ centered at the origin. Let $V$ be the set of generated vertices and $V(B)$ the number of vertices in a given measurable subset $B \subseteq \mathcal S_{d,n}$.
    
    \item Each vertex $v \in V$ is independently assigned a community label, denoted $\sigma^*(v)$, with $\P(\sigma^*(v) = i) = \pi_i > 0$ for each $i \in Z$. We write $V_i:=\{v \in V:\, \sigma^*(v)=i\}$ for the subset of Poisson points with community label $i$.
    
    \item Conditioned on the locations and community labels, edge weights are generated independently between each pair of vertices. Let $\overline{P}(y) := P(y/(\log n)^{1/d})$ and $\overline{p}(\cdot; y)$ be the density of $\overline{P}(y)$. For each pair of vertices $\{u, v\}$ where $u \neq v$, the edge weight $X_{uv}$ is sampled from
    \begin{equation*}
        X_{uv} \sim \overline{P}_{\sigma^*(u)\sigma^*(v) }(\|u-v\|),
    \end{equation*}
    where $\|\cdot\|$ denotes the Euclidean toroidal metric. We denote the observed edge weight as $x_{uv}$. We assume that 
    \[    r := \inf\{y > 0 : \exists x \in \mathbb X \text{ such that } P_{ab}(y') = \delta_x \text{ for all }  y' > y\text{ and } a,b \in Z\} < \infty,
    \] where $\delta_x$ is Dirac measure in $x \in \mathbb X$ and call $r (\log n)^{1/d}$ the \textit{visibility radius}. In addition, we say that $u$ is \textit{visible} to $v$ if $\|u - v\| \leq r (\log n)^{1/d}$, which we denote $u \visible v$. 
\end{enumerate}
\end{definition}
Our goal is to achieve exact recovery, which is a particular notion of recovering the community labels. As discussed in \cite{gaudio2025incguan}, note that if there are symmetries in the prior probabilities and edge weight distributions, then it is only possible to recover the correct labeling up to a permutation. Consequently, \cite{gaudio2025incguan} introduced the definition of a permissible relabeling to account for such symmetries, which we modify to account for distance-dependent edge weight distributions. 

\begin{definition}[Permissible Relabeling]
\label{def:permissible}
A permutation of communities $\omega: Z \to Z$ is called permissible if $\pi_i = \pi_{\omega(i)}$ and $P_{ij}(y) = P_{\omega(i)\omega(j)}(y)$ for all $i, j \in Z$ and almost all $y \in [0, r]$. We denote the set of permissible relabelings as $\Omega_{\pi, P}$.
\end{definition}
Now, we define exact recovery. Let $\sigma^*_n$ be the true labeling and $\tilde{\sigma}_n$ be an estimated labeling. We define the agreement of $\tilde{\sigma}_n$ and $\sigma^*_n$ as
\begin{equation*}
    A(\tilde{\sigma}_n, \sigma^*_n) = \frac{1}{|V|} \max_{\omega \in \Omega_{\pi, P}} \sum_{u \in V} \1\{\tilde{\sigma}_n(u) = \omega \circ \sigma^*_n(u)\}
\end{equation*}
which is the proportion of vertices that $\tilde{\sigma}_n$ labels correctly, up to a permissible relabeling. Then, we say that $\tilde{\sigma}_n$ achieves
\begin{itemize}
    \item \textit{exact recovery} if $\lim_{n \to \infty} \P(A(\tilde{\sigma}_n, \sigma^*_n) = 1) = 1$,
    \item \textit{almost-exact recovery} if $\lim_{n \to \infty} \P(A(\tilde{\sigma}_n, \sigma^*_n) \geq 1 - \epsilon) = 1$, for all $\epsilon > 0$.
\end{itemize}

To quantify the distinguishability between communities, we define the
\textit{Chernoff--Hellinger (CH) divergence} \citep{abbe_2015_sbm_limits}.
For $a\in Z$ and $y\in[0,r]$, let
\[
P_a(\,\cdot\,;y)
:=
\bigl(P_{ac}(\,\cdot\,;y)\bigr)_{c\in Z}
\]
denote the vector of edge-weight distributions from a vertex in
community $a$ to vertices in the respective communities at distance
$y$. For $t\in(0,1)$ and probability measures $P,Q\ll\mu$ with
densities $p,q$, respectively, define
\begin{equation}
\label{eq:def_phi}
\phi_t(P,Q)
\coloneqq
\int_{\mathbb X} p(x)^t q(x)^{1-t}\,\mu(dx).
\end{equation}
By the weighted arithmetic--geometric mean inequality,
\(
p(x)^tq(x)^{1-t}
\leq tp(x)+(1-t)q(x),
\)
and hence $\phi_t(P,Q)\leq 1$. Moreover,
 $\phi_t(P,Q)=1$ if and only if $P=Q$. 

We define the CH divergence between communities $a,b\in Z$ by
\begin{equation}
\label{eq:def_CH}
D_+(P_a\|P_b)
\coloneqq
1-
\inf_{t\in(0,1)}
\int_0^r
\frac{d y^{d-1}}{r^d}
\sum_{c\in Z}\pi_c\,
\phi_t\bigl(P_{ac}(y),P_{bc}(y)\bigr)
\,dy.
\end{equation}
Thus, $D_+(P_a\|P_b)$ measures the average distinguishability between
communities $a$ and $b$, where the edge-weight distributions are
averaged over the distance distribution and over the community label
of the neighboring vertex according to the prior $\pi$.

For $a\neq b$ and $I\subseteq[0,r]$, define the set of
\textit{witness communities}
\[
W_{ab}(I)
\coloneqq
\left\{
w\in Z:
\inf_{t\in(0,1)}
\phi_t\bigl(P_{aw}(y),P_{bw}(y)\bigr)<1
\text{ for a.e. }y\in I
\right\}.
\]
Thus, $w\in W_{ab}(I)$ precisely when the edge-weight distributions
from communities $a$ and $b$ to community $w$ are distinguishable for
almost every distance in $I$.

\begin{assumption}[Identifiability of pairwise distributions]
\label{ass:identifiability}Assume that 
\begin{equation}
\label{eq: conn constraint}
\lambda \nu_d \max_{0 \le l \le r }\min_{a \neq b}\max_{l \leq s \leq r} 
l^d \sum_{w \in W_{ab}([s-l, s])} \pi_w
>
\begin{cases}
1, & d \ge 2, \\
2, & d = 1,
\end{cases}
\end{equation}
where $\nu_d$ denotes $d$-dimensional volume of the unit ball in $\mathbb R^d$.

Intuitively, this condition requires that there exists a common interval length $l \in [0,r]$ such that for every pair of communities, there exists a distance interval of length $l$ on which the corresponding witness communities provide enough informative vertices to form a connected propagation structure.
\end{assumption}

Let $\Omega_P$ be the set of all permutations $\omega:Z\to Z $ such that $P_{ij}(y)=P_{\omega(i)\omega(j)}(y)
\text{ for all }i,j\in Z
\text{ and a.e. }y\in[0,r]$.
Intuitively, the following assumption says that at sufficiently short distances, assigning two vertices from the same true community to two different candidate communities produces a uniformly distinguishable pair distribution.
\begin{assumption}[Initialization separation]
\label{ass:init-identifiability}
There exist $r_{\mathrm{init}}\in(0,r]$ and
$\Phi_{\mathrm{init}}\in(0,1)$ such that, for every
$a,b,c\in Z$ with $b\neq c$,
\[
\phi_{1/2}\left(
P_{bc}(y),P_{aa}(y)
\right)
\leq
\Phi_{\mathrm{init}}
\]
for a.e. $y\in[0,r_{\mathrm{init}}]$.
\end{assumption}

\begin{assumption}[Bounded log-likelihood difference]
\label{ass:bounded-ll}
For any communities $a, b, a', b' \in {Z}$, there exists~$\rho > 0$ such that for all~$x \in \mathbb X$ and~$y \in [0,r]$,
    \begin{equation}
         \bigg|\log \frac{{p}_{ab}(x; y)}{{p}_{a'b'}(x; y)}\bigg| < \rho.
    \end{equation}
\end{assumption}

The following theorem is our main result. 
\begin{theorem}[Achievability]
\label{theorem: achiev}
Let Assumptions \ref{ass:identifiability}, \ref{ass:init-identifiability} and \ref{ass:bounded-ll} hold and suppose that 
\begin{align}
\lambda \nu_d r^d \min_{a \neq b} D_+({P}_a\| {P}_b) > 1. \label{ass:D}
\end{align}
Then Algorithms \ref{alg:exact-recovery} and \ref{alg:phase2_block} presented below achieve exact recovery in polynomial time (in $n$).
\end{theorem}
From \cite[Theorem 1]{gaudio2026jan} we know that Condition \eqref{ass:D} is a necessary condition for exact recovery (information-theoretic threshold). An important special case (considered in \cite[Theorem 2]{gaudio2026jan}) is the situation where for every pair $(a,b)\in Z^2$ with $a \neq b$ and every $w \in Z$ it holds that $P_{aw}(y)\neq P_{bw}(y)$ for almost all $y \in [0,r]$. In this case, we have that $W_{ab}([0,r])=Z$ for every $a \neq b$, and Assumption \ref{ass:identifiability} reduces to 
\begin{align*}
\lambda \nu_d r^d
>
\begin{cases}
1, & d \ge 2, \\
2, & d = 1,
\end{cases}
\end{align*}
Thus, in this case Assumption \ref{ass:identifiability} becomes trivial in dimension $d \ge 2$ whenever Condition \eqref{ass:D} is met. See Figure~\ref{fig: new constraint} for an illustration of different parameter regimes. 
\begin{figure}[H]
    \centering
    \includegraphics[
        width=\linewidth
    ]{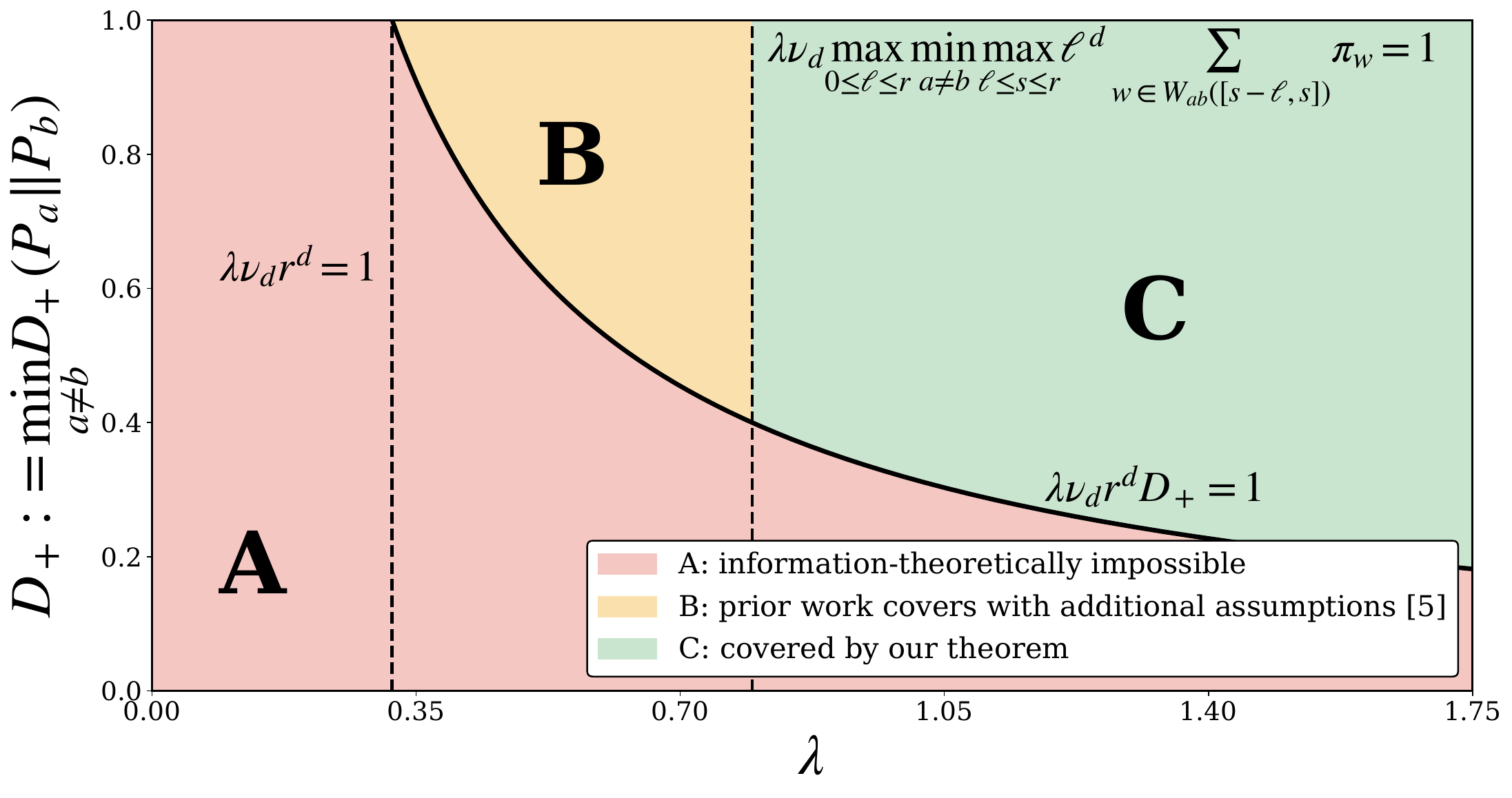}

\caption{Illustration of Theorem \ref{theorem: achiev}}
\label{fig: new constraint}
\end{figure}
\noindent
For $d\geq 2$, Region A is information-theoretically impossible, since
$\lambda\nu_d r^d D_+ < 1$ \cite[Theorem~2.1]{gaudio2025incguan}.
Region C is covered by Theorem~\ref{theorem: achiev},
where both the refinement condition $\lambda \nu_d r^d D_+ > 1$ and the witness-connectivity condition $\lambda \nu_d \pi_{\mathrm{wit}}L^d > 1$ hold. Region B is the remaining intermediate regime where exact recovery is achieved by~\cite{gaudio2026jan} under \cite[Assumptions 1 and 3]{gaudio2025incguan}.
The solid curve corresponds to the information-theoretic threshold $\lambda\nu_d r^d D_+=1$. The left dashed vertical line corresponds to $\lambda\nu_d r^d=1$, the visibility-graph connectivity threshold. The right dashed vertical line corresponds to the additional connectivity threshold $\lambda\nu_d \max_{0\leq l\leq r} \min_{a\neq b} \max_{l\leq s\leq r}l^d\sum_{w\in W_{ab}([s-l,s])}\pi_w =1$. The solid curve and the right dashed line intersect at a critical value of $D_+$. Below it, the information-theoretic condition already implies connectivity, so the two thresholds coincide. Above it, Region B lies between the information-theoretic and connectivity thresholds.


\section{Exact Recovery Algorithms}
Next we present our algorithms for exact recovery. For $u \in \mathcal S_{d,n}$ let $N(u):=\{v \in S_{d,n}:u~\visible~v, u\neq v\}$ be the {\em neighborhood} of $u$.

\begin{definition}[Affinity information]\label{def:affinity_information}
    Let $u \in \mathcal S_{d,n}$, $S \subseteq V \cap N(u)$ and let $\hat \sigma:S \to Z$ be a labeling on $S$. We call
    \begin{align}
        \mathcal{I}_S(u) \coloneqq \min_{a \neq b} \sup_{t \in (0,1)} \sum_{v \in S} - \log  \phi_t \left(\bar P_{a \hat \sigma(v)}(\|u-v\|), \bar P_{b\hat\sigma(v)}(\|u-v\|)\right).\label{eq:def_I}
    \end{align}
    the \textit{affinity information} of $u$ on the \textit{source} $S$ for a labeling $\widehat{\sigma}$.
\end{definition}

We now give an overview of our first, discretization-free Algorithm \ref{alg:exact-recovery}. We  fix the initial set
\begin{align}
    B_{\mathrm{init}} \coloneqq [-r(\log n)^{1/d},r(\log n)^{1/d}]^d.
\end{align}
As in \cite{gaudio2026jan}, we use a Maximum A Posteriori (MAP) estimator on $B_{\mathrm{init}}$, which maximizes the posterior likelihood of the observed graph. Since $k$ is fixed and $B_{\mathrm{init}}$ contains $O(\log n)$ vertices with high probability, exhaustive MAP over the initial set requires
\(
k^{O(\log n)} = n^{O(1)}
\)
evaluations and is therefore polynomial in $n$.

Second, we label all vertices outside the initial set via a propagation procedure. Given that a subset $\hat V\subseteq V$ has already been labeled, we take the node that maximizes the affinity information $\mathcal I_{\hat V \cap N(u)}(u)$ among all nodes $u \in V \setminus \hat V$. We then choose its label such that likelihood of the observed graph restricted to the vertex set $\hat V \cup \{u\}$ is maximized, and then update $\hat V$. We repeat this procedure until all of the nodes in $V$ are labeled. Ignoring vertex selection, propagating all vertices requires
$O(n\log n)$ operations. A naïve implementation of
the selection rule requires $O(n^2\log n)$ operations. By maintaining
the affinity scores incrementally in a priority queue, this can be
reduced to $O(n(\log n)^2)$. 

Third, we implement the same refinement as \cite{gaudio2026jan} to obtain exact recovery by iterating through all nodes, with complexity $O(n \log n)$. Given the estimated labeling from the initialization and propagation $\hat{\sigma}$, we compute $\tilde{\sigma}(v)$ for each vertex $v \in  V$ by choosing the community that maximizes its posterior likelihood given $\hat{\sigma}$ and $G = ({V}, E)$,
\[\tilde{\sigma}(u)
=
\underset{a \in Z}{\arg\max}
\sum_{\substack{v\in V:\,
u\visible v,\,
\hat\sigma(v)\neq *,\,
u\neq v}}
\log\left(
\bar p_{a,\hat\sigma(v)}
(x_{uv};\|u-v\|)
\right).\]
We next present pseudocode for Algorithm \ref{alg:exact-recovery}. The subroutines~\texttt{Maximum a Posteriori}, \texttt{Propagate} and \texttt{Refine} are presented in the Appendix.
\begingroup
\renewcommand{\thealgorithm}{A}
\begin{algorithm}[H]
\caption{Discretization-free Exact Recovery}
\label{alg:exact-recovery}
\begin{algorithmic}[1]
\Statex \textbf{Input:} $G\sim \mathrm{GHCM}(\lambda,n,r,\pi,P(\cdot),d)$.
\Statex \textbf{Output:} An estimated community labeling $\widetilde\sigma: V\to Z$.
\Statex
\State \textbf{Phase I (Initialization):}
\State  Let $\hat V:= V \cap B_{\mathrm{init}}$. 
\State Apply \texttt{Maximum a Posteriori} (Algorithm \ref{alg:map}) on input $(G, \hat V)$ to obtain a labeling $\hat{\sigma}$ on $\hat V$.
\Statex
\State \textbf{Phase II (Propagation):}
\While{$V\setminus \hat V \neq \emptyset$}
    \State Let $u \coloneqq \argmax_{v\in V \setminus \hat V} \mathcal I_{N(v) \cap \hat V}(v)$.
   \State  Apply \texttt{Propagate} (Algorithm \ref{alg:propagate}) on input $(G,\hat V,\{u\},\hat\sigma_{\hat V})$ to obtain a labeling $\hat\sigma(u)$.
\State $\hat V\leftarrow\hat V\cup\{u\}$.
\EndWhile

\Statex
\State \textbf{Phase III (Refinement):}
\For{$u\in V$}
\State Apply \texttt{Refine} (Algorithm \ref{alg:refine}) on input $(G, u,\hat \sigma)$ to obtain a labeling $\tilde{\sigma}(u)$.
\EndFor
\end{algorithmic}
\end{algorithm}
\endgroup

We now present our block-based Algorithm \ref{alg:phase2_block}, generalizing the algorithm from \cite{gaudio2026jan} (originating from \cite{gaudio2024exactcommunityrecoverygeometric}). First we need some more definitions. Let
\begin{align}
    \ell \coloneqq \argmax_{0 \le l \le r }\min_{a \neq b}\max_{l\le s\le r} 
l^d \sum_{w \in W_{ab}([s-l,s])} \pi_w,\qquad \pi_{\mathrm{wit}}:=\min_{a \neq b}\max_{\ell\le s \le r} 
\sum_{w \in W_{ab}([s-\ell,s])} \pi_w,\label{def:ell}
\end{align}
and note that $\lambda \nu_d \ell^d \pi_{\mathrm{wit}}$ is the left-hand side of \eqref{eq: conn constraint}. For every $a\neq b$, fix $s_{ab}\in[\ell,r]$ such that
\[
\sum_{w\in W_{ab}([s_{ab}-\ell,s_{ab}])}\pi_w
\geq \pi_{\mathrm{wit}},
\]
and set
\[
I_{ab}:=[s_{ab}-\ell,s_{ab}].
\]
Thus, $|I_{ab}|=\ell$ and
\[
\sum_{w\in W_{ab}(I_{ab})}\pi_w
\geq \pi_{\mathrm{wit}}
\qquad\text{for every }a\neq b.
\]

Next we describe how we partition $\mathcal S_{d, n}$. If $d=1$, take
\begin{equation}\label{eq:condition_chi_0_d1}
0<\chi < \frac{1}{2}\Big(1-\frac{1}{\lambda\pi_{\mathrm{wit}}\ell}\Big).
\end{equation}
If $d\geq 2$, take $\chi > 0$ such that    \begin{equation}\label{eq:condition_chi_0}
       \lambda \pi_{\mathrm{wit}} \ell^d \Big( \nu_d \Big(
            1-\frac{3\sqrt d}{2}\chi^{1/d}
            \Big)^d- \chi \Big) > 1
        \quad \text{and} \quad
            1 - \frac{3\sqrt d}{2}\chi^{1/d} > 0.
    \end{equation}
The existence of $\chi$ is guaranteed by Lemma \ref{lemma:condition_chi_0}. 
We discretize $\mathcal S_{d, n}$ into a family $\mathcal P_{\chi}$ of cubes of volume $\ell^d \chi \log n$, which we call \textit{blocks}. For $s>0$ we call two blocks 
$B_i$ and $B_j$ $s$-\textit{close} if 
    \begin{equation}
        \sup_{u\in B_i, v\in B_j} \|u - v\| \le s (\log n)^{1/d},
    \end{equation}
where  $\| \cdot \|$ is the toroidal metric on $\mathcal S_{d,n}$. For $B, B' \in \mathcal P_{\chi}$ we write $B \sim B'$ if $B$ and $B'$ are $\ell$-close with $\ell$ given at \eqref{def:ell}, and $B \visible B'$ if they are $r$-close. Let $K_d$ denote the number of blocks that are $\ell$-close to a given block $B$, excluding the block itself. Given a graph $G$ on $\mathcal S_{d,n}$ and a partition of $\mathcal S_{d,n}$ into blocks of volume $v(n)$, we follow \cite[Definition 5]{gaudio2026jan} and define the $(v(n),c(n))$-block visibility graph of $G$ as the graph $H=(V^\dagger, E^\dagger)$, where $V^\dagger := \{i \in [n/v(n)] : V(B_i) \geq c(n)\}$ is the set of all blocks with at least $c(n)$ vertices and $E^\dagger := \{\{i, j\}: i, j \in V^\dagger, B_i \sim B_j\}$ is the set of all pairs of blocks in $V^\dagger$ that are mutually $\ell$-close. Let $\rho>0$ satisfy Assumption \ref{ass:bounded-ll}, define
\begin{align}\label{eq:Kr}
\delta_0:=\frac{(\lambda \pi_{\mathrm{wit}}\ell^d \chi K_d-1)^2}{8 \lambda \pi_{\mathrm{wit}}\ell^d \chi K_d^2},\quad
\beta :=\frac{1}{4\rho}(\lambda \nu_d r^d \min_{a\neq b}D_+( P_a \|  P_b)-1),\quad  K_r \coloneqq  
    \left\lceil
    \frac{\nu_d(r+\sqrt d \ell\chi^{1/d})^d}{\ell^d\chi}
    \right\rceil 
\end{align}
  and fix \begin{align}\label{eq:condition_delta}
       0 < \delta < \min\left\{\delta_0, \frac{\beta}{K_r}\right\}
    .\end{align}
We call a block $B$ $\delta$-occupied block if $V(B) \ge \delta \log n$. The parameters $\chi$ and $\delta$ are chosen so that the blocks contain enough nodes of each community and so that the $(\ell^d \chi \log n, \delta \log n)$-block visibility graph $H=(V^\dagger, E^\dagger)$ of $\delta$-occupied blocks is connected with high probability. 

Given a labeling $\hat\sigma$, we call a $\delta$-occupied block $B$
\emph{$\hat\sigma$-certified} if
\begin{align}
\min_{a\neq b}
\sum_{w\in W_{ab}(I_{ab})}
\left|
\left\{
v\in V(B):\hat\sigma(v)=w
\right\}
\right|
\ge
\frac12\delta\pi_{\mathrm{wit}}\log n.
\label{def:certified}
\end{align}
If $\delta$ and $\hat \sigma$ are clear from the context, let $\hat V_{\mathrm{cert}}^\dagger \subseteq V^\dagger$ denote the subset of $\hat \sigma$-certified blocks. 

\medskip
We next give an overview of our second, block-based algorithm, which generalizes the algorithm given in \cite{gaudio2026jan}. We start performing a Maximum A Posteriori estimate on the initial set $B_{\mathrm{init}}$, which is identical to Algorithm \ref{alg:exact-recovery}. Then we divide $\mathcal S_{d,n}$ into blocks of volume $\ell^d \chi \log n$ and construct the $(\ell^d \chi \log n,\delta \log n)$-visibility graph $H=(V^\dagger,E^\dagger)$ of $\delta$-occupied blocks. Given that $H$ is connected, we determine the subset of all $\hat \sigma$-certified blocks. We then choose a $\delta$-occupied block $B$ which is not fully contained in $B_{\mathrm{init}}$ and $\ell$-close to a $\hat \sigma$-certified block. We apply Propagate (Algorithm \ref{alg:propagate}) on input $(G, \bigcup_{k \in \hat V_{\mathrm{cert}}^\dagger} V \cap B_k, V \cap B_i,\hat \sigma)$ to obtain a labeling $\hat \sigma$ on the vertex set of $B$. We then update $\hat V^\dagger$ and $\hat V_{\mathrm{cert}}^\dagger$ and apply Propagate on the next $\delta$-occupied block which is $\ell$-close to a $\hat \sigma$-certified block. We will show that this procedure allows us to label all blocks in $H$ with high probability. Such a technique of \textit{propagation} was first introduced in 2024 by \cite{gaudio2024exactcommunityrecoverygeometric}. The difference to the algorithm in \cite{gaudio2026jan} is that in the propagation step we do not only use the estimated labeling in one parent block, but from all labeled blocks intersecting the visibility range of a given block. This is necessary, since we work under the weaker Assumption \ref{ass:identifiability}, which only guarantees that a given pair of communities can be distinguished in a small interval. Moreover, we need to use nodes of all communities in $Z$, and cannot restrict ourselves to nodes with the most present community label. Finally, exact recovery is obtained by a refinement step that is identical to above.

\begingroup
\renewcommand{\thealgorithm}{B}
\begin{algorithm}[H]
    \caption{Block-based Exact Recovery}
    \label{alg:phase2_block}
    \begin{algorithmic}[1]
\Statex \textbf{Input:} $G\sim \mathrm{GHCM}(\lambda,n,r,\pi,P(\cdot),d)$.
\Statex \textbf{Output:} An estimated community labeling $\widetilde\sigma: V\to Z$.

\medskip
\State \textbf{Phase I:}
\State  Let $\hat V:= V \cap B_{\mathrm{init}}$. 
\State Apply \texttt{Maximum a Posteriori} (Algorithm \ref{alg:map}) on input $(G, \hat V)$ to obtain a labeling $\hat{\sigma}$ on $\hat V$.

\medskip
\State \textbf{Phase II:}
\State Take $\chi$ satisfying \eqref{eq:condition_chi_0_d1} if $d=1$, or \eqref{eq:condition_chi_0} if $d\geq 2$, and $\delta>0$ satisfying \eqref{eq:condition_delta}.
\State Partition $\mathcal S_{d,n}$ into $n / (\ell^d \chi \log n)$ blocks of volume $\ell^d\chi\log n$.
\State Let $B_i$ be the $i$-th block for $i \in [\lceil n/(\ell^d \chi \log n)\rceil]$.
\State Construct the $(\ell^d \chi \log n,\delta \log n)$-visibility graph $H=(V^\dagger,E^\dagger)$ of $\delta$-occupied blocks.
\State Let $\hat V^\dagger:=\{i \in V^\dagger:\,B_i \subseteq B_{\mathrm{init}} \}$ and let $\hat V_{\mathrm{cert}}^\dagger \subseteq \hat V^\dagger$ be the subset of all $\hat \sigma$-certified blocks.
\While {$\{k \in V^\dagger \setminus \hat V^\dagger:\,\exists j \in \hat V_{\mathrm{cert}}^\dagger:\,B_j \sim B_k\} \neq \emptyset$}
\State Let $i:=\min \{k \in V^\dagger \setminus \hat V^\dagger:\,\exists j \in \hat V_{\mathrm{cert}}^\dagger:\,B_j \sim B_k\}$.

\State Apply \texttt{Propagate} on input
$(
G,\,
\bigcup_{k\in\hat V_{\mathrm{cert}}^\dagger}(V\cap B_k),\,
(V\cap B_i)\setminus\hat V,\,
\hat\sigma
)$
to obtain labels on $(V\cap B_i)\setminus\hat V$.
\State $\hat V\leftarrow\hat V\cup(V\cap B_i)$.
\State $\hat V^\dagger\leftarrow\hat V^\dagger\cup\{i\}$.
 \If{$B_i$ is $\hat \sigma$-certified}
$\hat V_{\mathrm{cert}}^\dagger \leftarrow \hat V_{\mathrm{cert}}^\dagger  \cup \{i\}$.
 \EndIf
\EndWhile
\For{$v \in V \setminus \bigcup_{i \in V^\dagger} (V \cap B_i)$}
    \State Set $\hat{\sigma}(v) = *$.
\EndFor
\medskip

\State \textbf{Phase III:}
\For{$v \in V$}
    \State Apply \texttt{Refine} (Algorithm \ref{alg:refine}) on input $(G, v, \hat{\sigma})$ to compute $\tilde{\sigma}(v)$.
\EndFor
    \end{algorithmic}
\end{algorithm} 
\endgroup

\section{Proof Sketch}
We provide a brief sketch of the argument showing that Algorithms \ref{alg:exact-recovery} and \ref{alg:phase2_block} achieve exact recovery. The detailed proof of Theorem \ref{theorem: achiev} can be found in the Appendix. 

\textbf{Initialization.} We adapt the proof of \cite{gaudio2026jan} to the weaker assumptions considered here. In particular, no initial sampling is required, as the random number of Poisson points in the initial set $B_{\mathrm{init}}$ is handled directly using the Mecke equation. This yields exact recovery on $B_{\mathrm{init}}$.

\textbf{Propagation.}
The proof proceeds by showing that errors cannot accumulate during the
propagation phase. The argument has three ingredients.

First, Lemma~\ref{lem:enough_info} shows that, with high probability,
every vertex encountered during propagation has sufficiently many
previously labeled neighbours. More precisely, for every propagated vertex $v$, with $S(v)$ denoting its set of previously labeled neighbours used for propagation, the available affinity information satisfies
\[
\mathcal I_{S(v)}(v)\ge c_2\log n,
\]
unless the number of errors in one of the relevant neighbourhoods already
exceeds a fixed constant \(m\).

Second, conditionally on all information revealed before \(v\) is
processed, Lemma~\ref{lemma:one_step_error_bound} gives
\[
\mathbb P\left(
\widehat\sigma(v)\notin
\{\omega^\star\circ\sigma^\star(v),*\}
\,\middle|\,\mathcal F_v
\right)
\le
\exp\{2\rho m-\mathcal I_{S(v)}(v)\}
\le e^{2\rho m}n^{-c_2}.
\]
Thus, as long as there are at most \(m\) previous errors in every relevant
neighbourhood, the probability of creating a new error is polynomially
small in \(n\).

Finally, we show that these rare errors cannot accumulate to more than
\(m\) errors in any neighbourhood. Fix a vertex \(w\) and consider the
vertices in \(N(w)\) in their propagation order. Although their error
indicators are dependent because the propagation order is adaptive, the
conditional error bound above implies, by iterating the tower property,
that for any \(m+1\) vertices \(v_1<\cdots<v_{m+1}\),
\[
\mathbb E\left[
\prod_{i=1}^{m+1}\xi_{v_i}\,\middle|\,V
\right]
\le
\bigl(e^{2\rho m}n^{-c_2}\bigr)^{m+1}.
\]
Since by Corollary \ref{cor:poisson_upper_concentration}, with probabiltiy $1-o(1)$, every visibility neighbourhood contains  \(O(\log n)\) vertices, a union bound over all \(m+1\)-subsets shows that the
probability of \(m+1\) errors in a fixed neighbourhood is
\(o(n^{-1})\), provided \(m\) is chosen sufficiently large. A further
union bound over the \(O(n)\) relevant neighbourhoods then gives an
\(o(1)\) probability that any neighbourhood contains more than \(m\)
errors.

More precisely, if \(v^\star\) denotes the first vertex for which the
desired error bound fails, then all vertices preceding \(v^\star\)
satisfy the error bound. Hence the \(m+1\) errors witnessing the failure
at \(v^\star\) are all among the auxiliary error variables controlled by
the preceding estimate. This yields the required contradiction.

Consequently, with high probability, throughout the entire propagation
procedure every neighbourhood contains at most \(m\) incorrectly
labelled vertices.

\textbf{Refinement.} The argument is analogous to the corresponding one in \cite{gaudio2026jan}, but we include it for completeness. The main observation is that, for any threshold $\beta>0$, $\delta$ can be chosen sufficiently small such that, for every vertex, the probability of making an error when applying Algorithm~\ref{alg:refine} is $o(1/n)$. A union bound over all vertices then yields Exact Recovery, and consequently Theorem~\ref{theorem: achiev}.

\section{Numerical Analysis}

Although Algorithms~\ref{alg:exact-recovery} and
\ref{alg:phase2_block} are proved under the same sufficient assumptions,
their propagation mechanisms are different.
In Region B, the witness-connectivity condition fails, so no choice of $\chi$ satisfying the block-construction condition in \eqref{eq:condition_chi_0} is possible.
The discretization-free propagation rule remains well defined, so 
we investigate numerically whether this additional connectivity condition also appears to be necessary.
We consider a model in dimension $d=2$, $r=3$ and three communities with prior $\pi \coloneqq \left(\tfrac12,\tfrac3{10},\tfrac15\right)$. For $p\in(0,1)$ and $a,b \in\{1,2,3\}$, we define the distance-dependent edge-weight distributions
\begin{equation}
\label{eq:numerical-model}
P_{ab}(y)\coloneqq 
\begin{cases}
\operatorname{Ber}(p),
& a=b,\quad y\in[0,r/3)\cup(2r/3,r],\\
\operatorname{Ber}(1-p),
& a\neq b,\quad y\in[0,r/3)\cup(2r/3,r],\\
\operatorname{Ber}(1/2),
& y\in[r/3,2r/3],
\end{cases}
\end{equation}
where $\operatorname{Ber}(q)$ denotes Bernoulli distribution with parameter $q\in [0,1]$. Thus, observations are informative at short and long distances, while the intermediate range carries no community information.
This simple model captures both generalizations considered in this work (i.e. allowing $P_{aw}=P_{bw}$ for some $w \in Z$ and distance-dependent distributions with informative and non-informative intervals).
The connectivity functions coincide on a non-trivial interval, and for each
pair $a\neq b$, only communities $a$ and $b$ are witnesses.
Consequently, $\ell=r/3$ and $\pi_{\mathrm{wit}}=\min_{a\neq b}\pi_a+\pi_b=1/2$. 
We define $D_+(p) \coloneqq \min_{a\neq b}D_+(P_a\|P_b) = \tfrac13\left(1-2\sqrt{p(1-p)}\right)$. The exact-recovery threshold is therefore
$\lambda\nu_d r^dD_+(p)=1$, the witness-connectivity threshold is $\lambda\nu_d\ell^d\pi_{\mathrm{wit}}=1$, so we get a non-empty Region~B as illustrated in Figure~\ref{fig: new constraint}.

For the experiment, we fix the scaling parameter to $n=10\,000$.
For each cell of the $(\lambda,D_+)$ phase diagram, we generate independent GHCM realizations, and apply Algorithm~\ref{alg:exact-recovery}.
To isolate the propagation phase, the initial box $V \cap B_{\mathrm{init}}$ is labeled according to the ground truth, this corresponds on average to about $3.3\%$ of the vertices. This can be viewed as conditioning on successful initialization.
We use an adaptive grid, with increased resolution near the exact-recovery
boundary. Cells intersecting the boundary are recursively subdivided up to
five levels. We use $20$ independent trials for cells along the boundary and $10$ trials elsewhere. Since the three prior probabilites are distinct, $\Omega_{\pi,P}=\{\mathrm{id}\}$ in this experiment,  where $\mathrm{id}:Z\to Z$ denotes the identity map. Each cell reports the empirical exact-recovery frequency after refinement,
namely $\tfrac1T\sum_{t=1}^{T}\mathds{1}\left\{\widehat{\sigma}_{(t)}=
\sigma^{*}_{(t)}\right\}$ where $T$ is the number of independent trials in that cell and $\widehat{\sigma}_{(t)}$ and $\sigma^{*}_{(t)}$ are the estimated and ground-truth labeling for trial $t$.

\begin{figure}[H]
    \centering
    \includegraphics[width=\linewidth]{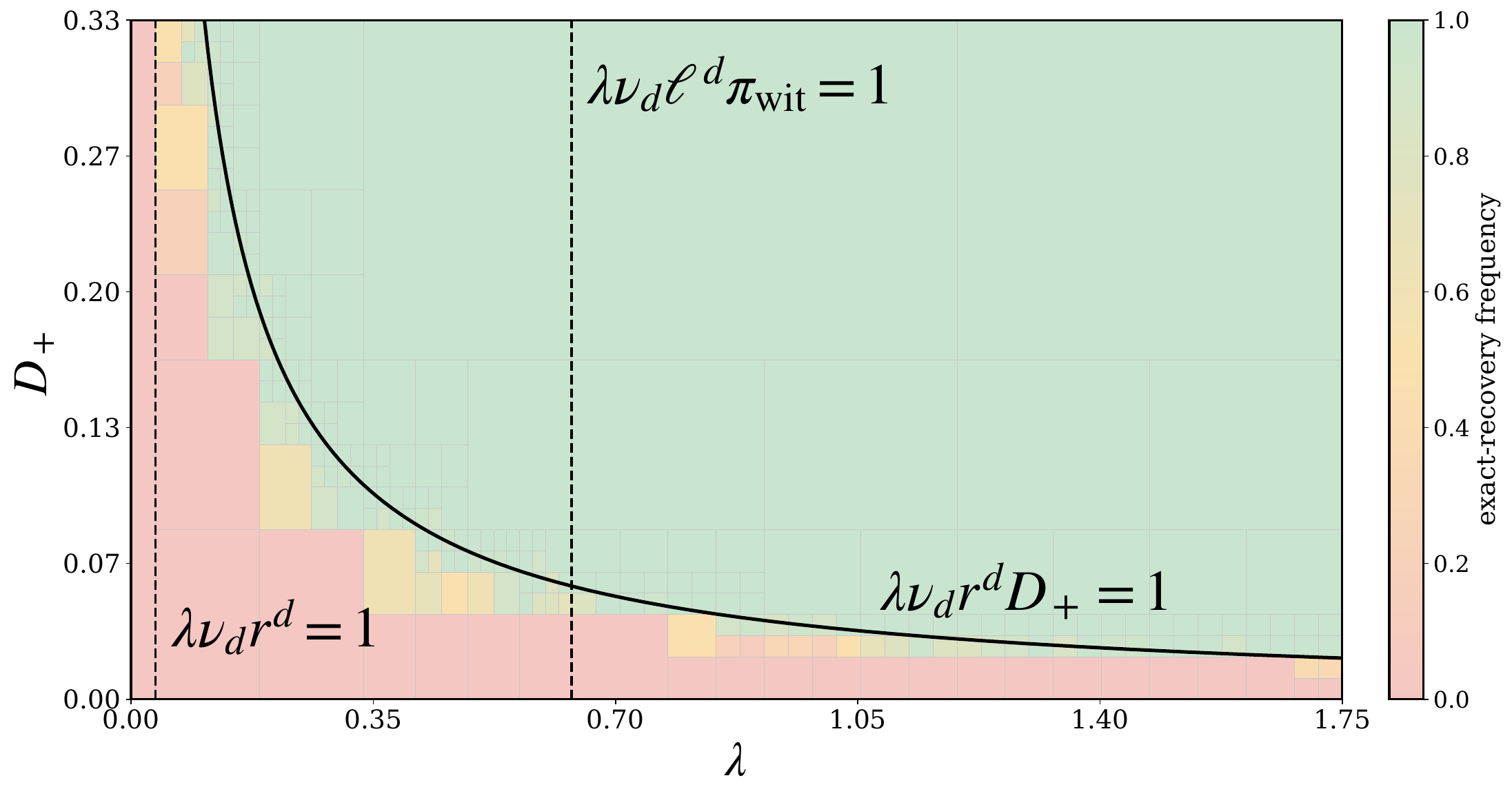}
    \caption{
    Empirical exact-recovery frequency of the discretization-free     Algorithm \ref{alg:exact-recovery} followed by refinement.
    The solid curve is the information-theoretic threshold
    $\lambda\nu_d r^dD_+=1$.
    The left dashed line is the visibility-graph connectivity threshold
    $\lambda\nu_d r^d=1$, while the right dashed line is the
    witness-connectivity threshold
    $\lambda\nu_d\ell^d\pi_{\mathrm{wit}}=1$ required by the block
    construction.
    The grid is refined near the information-theoretic boundary.
    }
    \label{fig:greedy-heatmap}
\end{figure}

Figure~\ref{fig:greedy-heatmap} shows that the empirical transition of
Algorithm \ref{alg:exact-recovery} occurs close to the information-theoretic boundary and extends well
inside Region~B. In particular, Algorithm \ref{alg:exact-recovery} achieves exact recovery in a regime where the
sufficient witness-connectivity condition used by the block construction is
violated. This suggests that the additional connectivity condition may be a limitation of the block-based propagation mechanism rather than an intrinsic barrier for discretization-free propagation.

\section{Discussion and Open Problems}

In this paper, we established exact recovery for two polynomial-time
algorithms under substantially weaker propagation assumptions than those
of \cite{gaudio2026jan}. Our analysis replaces the global distinctness
conditions used in previous propagation arguments by the weaker
witness-connectivity condition of Assumption~\ref{ass:identifiability}.
We establish this guarantee both for a discretization-free greedy
propagation procedure and for a block-based algorithm that aggregates
information across multiple previously labeled blocks.

This condition yields two regimes, illustrated in Figure~\ref{fig: new constraint}. When the information-theoretic condition implies witness connectivity, our algorithm achieves exact recovery down to the information-theoretic threshold, which is known to be necessary. In the remaining intermediate regime, the information-theoretic condition holds but our witness-connectivity condition does not. Whether polynomial-time exact recovery is possible in this regime without the assumptions of \cite{gaudio2026jan} remains open. Moreover, our numerical experiments suggest that the discretization-free algorithm itself succeeds beyond the range covered by our analysis. Proving exact recovery for this algorithm beyond our current theoretical guarantee is therefore an important open problem.

A second open direction is to remove Assumption~\ref{ass:bounded-ll}. This would allow edge-weight distributions with different supports, but creates difficulties because propagation relies on estimated labels. A single incorrectly labeled vertex can, in models with vanishing densities, assign zero likelihood to the correct community and thereby disrupt subsequent propagation. Addressing this issue likely requires a more robust propagation rule or an explicit treatment of uncertainty in previously estimated labels. While analogous boundedness assumptions can be avoided in the Gaussian setting of \cite{gaudio2025incguan}, the general case remains open.

\section*{Acknowledgment}
Part of this work was carried out during an internship of RS at Leiden University in summer 2026. Financial support from the Drs.~Kuikenga Fonds voor Mathematici (grant no.~W253098-1-035) is gratefully acknowledged.

\section*{AI statement}
OpenAI's ChatGPT was used to explore proof ideas and to assist with language polishing. All mathematical arguments and details were independently verified by the authors, who take full responsibility for their correctness.

\bibliographystyle{plainnat}
\bibliography{reference}
\newpage
\appendix


\section{Initialization, Propagagtion and Refinement algorithms}
\setcounter{algorithm}{0}
\renewcommand{\thealgorithm}{\arabic{algorithm}}

\begin{algorithm}[H]
\caption{Maximum a Posteriori}
\label{alg:map}
\begin{algorithmic}[1]
\Statex \textbf{Input:} Graph $G$ and vertex set $S\subset V$.
\Statex \textbf{Output:} An estimated labeling $\widehat\sigma_S:S\to Z$.

\State Set
\[
    \widehat\sigma_S
    =
    \argmax_{\sigma:S\to Z}
    \mathbb P(\sigma_S^*=\sigma\mid G).
\]
\end{algorithmic}
\end{algorithm}

\begin{algorithm}[H]
\caption{Propagate}
\label{alg:propagate}
\begin{algorithmic}[1]
\Statex \textbf{Input:} Graph $G$,
$S, S'$ disjoint mutually vertex sets and a labeling $\widehat\sigma_{S}:S\to Z$.
\Statex \textbf{Output:} An estimated labeling $\widehat\sigma_{S'}: S'\to Z$.
\State Set
    \[
        \widehat\sigma_{S'}(u)
        =
        \argmax_{b\in Z}
        \sum_{v\in S: u \visible v}
        \log p_{\hat\sigma_S(v)b}(x_{uv};\|u-v\|).
    \]

\end{algorithmic}
\end{algorithm}

\begin{algorithm}[H]
\caption{Refine}
\label{alg:refine}
\begin{algorithmic}[1]
\Statex\textbf{Input:} $G\sim \mathrm{GHCM}(\lambda,n,r,\pi,P(\cdot),d)$, a vertex $u\in  V$,
and a labeling $\widehat\sigma: V\to Z\sqcup\{*\}$.
\Statex \textbf{Output:} An estimated labeling $\widetilde\sigma :  V \to Z$.

\State Set
\[
    \widetilde\sigma(u)
    =
    \argmax_{a\in Z}
    \sum_{v\in  V:u\visible v, \hat \sigma(v) \neq *}
    \log\left(p_{a,\widehat\sigma(v)}(x_{uv};\|u-v\|)\right).
\]

\end{algorithmic}
\end{algorithm}

\section{Preliminaries}
\label{app: appendix}

The following Mecke equation is a standard tool in the analysis of Poisson processes. For a proof see \cite[Theorem 4.1]{lp}.

\begin{proposition}[Mecke equation] \label{thm:mecke}
    For any measurable $f\ge 0$,
    \begin{equation}  
    \E \Big[\sum _{v \in V} f(v,V)\Big] = \lambda \int_{\mathcal S_{d,n}} \E[ f(v,V\cup \{v\})]
\,\mathrm d v.
\end{equation}
\end{proposition}

The following lemma provides a standard concentration inequalities. The statement is taken from \cite{mitzenmacher} and its proof can be found therein.

\begin{lemma}[Poisson Chernoff bound]
\label{lemma:poisson_concentration}
Let $X\sim \operatorname{Poisson}(\lambda)$ for some $\lambda>0$. Then we have
\begin{align*}
    &\P(X  \ge x) \le \frac{(e \lambda)^x e^{-\lambda}}{x^x},\qquad x > \lambda,\\
    &\P(X  \le x) \le \frac{(e \lambda)^x e^{-\lambda}}{x^x},\qquad x < \lambda.
\end{align*}
\end{lemma}

Lemma \ref{lemma:poisson_concentration} allows us to find an upper bound for the number of nodes in a block.

\begin{corollary}\label{cor:poisson_upper_concentration}
    There exists a constant $\Delta > 0$ such that, with high probability, for all blocks $B_i$, we have $V(B_i) \leq \Delta \log n$.
\end{corollary}
\begin{proof}
    Note that $V(B_i)$ is Poisson($\lambda \chi \ell^d \log n$)-distributed. By Lemma~\ref{lemma:poisson_concentration} and a union bound, we have for all $\Delta>\lambda \chi \ell^d$, 
    $$
    \mathbb P\left(\bigcup_{i=1}^{\lceil n/(\ell^d \chi \log n)\rceil}  \{V(B_i) > \Delta \log n\}\right) \le 
    O\Big(\frac{n}{\log n}\Big) n^{- \lambda \chi \ell^d} \Big(\frac{e \lambda \chi \ell^d}{\Delta} \Big)^{\Delta \log n}.
    $$
    Since $\Delta \log\Big(\frac{e \lambda \chi \ell^d}{\Delta} \Big) \to -\infty$ as $\Delta \to \infty$, the assertion follows.
\end{proof}

The following lemma will be useful in the proof of exact recovery of the initialization step and almost-exact recovery of the propagation step. It generalizes {\cite[Lemma 5]{gaudio2026jan}}.

\begin{lemma}[Informative interval]
\label{lem: informative interval}
Let $a,b,a',b'\in Z$ and $t\in(0,1)$. Suppose that there exists an
interval $I\subseteq[0,r]$ of positive length such that
\[
P_{ab}(y)\neq P_{a'b'}(y)
\qquad\text{for a.e. }y\in I.
\]
Then, for every $\varepsilon>0$, there exists $\gamma\in(0,1)$ such that
\[
\left|
\left\{
y\in I:
\phi_t\bigl(P_{ab}(y),P_{a'b'}(y)\bigr)
\geq 1-\gamma
\right\}
\right|
<\varepsilon |I|.
\]
\end{lemma}

\begin{proof}
    We skip the proof since it is almost identical to the proof of \cite[Lemma 5]{gaudio2026jan}.
\end{proof}

Let $B_r:=\{x \in \R^d:\,\|x\|\le r\}$. For any interval $I \subset [0,r]$ define the annulus
\begin{align}
\mathcal A_{I}:=\Big\{u \in B_r:\,\|u\|\in I \Big\}. \label{def:Aab}
\end{align}

The following corollary is a direct consequence of Lemma \ref{lem: informative interval}. It is the analogue of \cite[Corollary 1]{gaudio2026jan}.

\begin{corollary} \label{cor:D}
    Let $a,b,a',b'\in Z$ and let $I\subseteq[0,r]$ be an interval of positive length such that
    \[
    P_{ab}(y)\neq P_{a'b'}(y)
    \qquad\text{for a.e. }y\in I.
    \]
    Let $D$ denote the distance from the origin of a point chosen uniformly at random from $\mathcal A_I$. Then, for every $\epsilon>0$ and $t\in(0,1)$, there exists $\gamma\in(0,1)$ such that
    \[
    \mathbb P\left(
    \phi_t\bigl(P_{ab}(D),P_{a'b'}(D)\bigr)
    \geq 1-\gamma
    \right)
    \leq \epsilon.
    \]
    The value of $\gamma$ depends only on $\epsilon$, $t$, and $I$.
\end{corollary}


\section{Exact recovery on the initial set}
In this section we show that the Maximum A Posteriori (MAP) estimator on the initial set achieves exact recovery. To get an efficient algorithm, we restrict to a logarithmic size and take $B_{\mathrm{init}} = [-r (\log n)^{1/d},r (\log n)^{1/d}]^d$ and $V_{\mathrm{init}}:=V \cap B_{\mathrm{init}}$. We follow the strategy from \cite{gaudio2025incguan, gaudio2026jan} and analyze a restricted Maximum Likelihood Estimator (MLE) rather than the MAP itself. We will show that the MLE given by
\begin{equation}
    \bar{\sigma} = \underset{\sigma \in X^{*}_0(\epsilon)}{\arg \max }\sum_{u \in {V_{\mathrm{init}}}}\sum_{v \in {V_{\mathrm{init}}}, v\neq u} \log \left( \bar{p}_{\sigma(u),\sigma(v)}(x_{uv}; \|u-v\|)\right)
\end{equation}
with
\begin{equation}
\label{eq: restr mle}
    X^*_0(\epsilon) := \left\{ \sigma : {V_{\mathrm{init}}} \rightarrow {Z} \ \Big| \ | \{u \in {V_{\mathrm{init}}} : \sigma(u) = j \} | \in \big( (\pi_j - \epsilon)|V_{\mathrm{init}}|, (\pi_j + \epsilon)|V_{\mathrm{init}}| \big) \quad \forall j \in {Z}   \right\}
\end{equation}
achieves exact recovery on $V_{\mathrm{init}}$ with high probability. Since the MAP is optimal, this implies that the MAP achieves exact recovery on $V_{\mathrm{init}}$ with high probability as well.

\begin{proposition}
\label{prop: initial ball}
Let ${G} \sim \mathrm{GHCM}( \lambda,n, r, \pi, P(\cdot), d)$ and let $\bar{\sigma}$ be the restricted MLE. Suppose that Assumptions~\ref{ass:identifiability}, \ref{ass:init-identifiability}  and \ref{ass:bounded-ll} hold. Then,
$\bar{\sigma}$  achieves exact recovery on $V_{\mathrm{init}}$ with high probability.
\end{proposition}

Our proof of Proposition \ref{prop: initial ball} is an adaptation of the proof of Lemma D.4 in \cite{gaudio2025incguan}.

\begin{lemma}
\label{lem:restr mle proof}
For any $0 < \epsilon < \min_{j \in Z}\pi_{j}/2$, it holds that $\mathbb{P}(\sigma^* \in X_0^*(\epsilon)) = 1 - o(1)$.
\end{lemma}
\begin{proof}
The number of vertices in the initial cube with community label $j$ is an independent Poisson random variable $N_j \sim \text{Poisson}(\lambda \pi_j (2r)^d \log n)$, and $|V_{\mathrm{init}}| = \sum_{j \in  Z} N_j \sim \text{Poisson}(\mu)$ where $\mu := \lambda (2r)^d \log n$. We show that $N_j \in ((\pi_j - \epsilon)|V_{\mathrm{init}}|, (\pi_j + \epsilon)|V_{\mathrm{init}}|)$ for all $j \in Z$ with high probability.

First we show that $|V_{\mathrm{init}}|$ concentrates around $\mu$. By Lemma \ref{lemma:poisson_concentration} and the elementary inequality $(e/(1+\epsilon))^{1+\epsilon}\le e^{1-\epsilon^2/4}$ for all $\epsilon \in (0,1)$,
\begin{align*}
\mathbb{P}\Big(||V_{\mathrm{init}}| - \mu| \geq \frac{\epsilon}{4}\mu\Big) &\leq \mathbb{P}\Big(|V_{\mathrm{init}}| \ge (1+\frac{\epsilon}{4} ) \mu\Big) +  \mathbb{P}\Big(|V_{\mathrm{init}}| \le (1-\frac{\epsilon}{4} ) \mu\Big)\\
&\le 2 e^{-\mu} \Big(\frac{e}{1+\frac \epsilon 4}\Big)^{(1+\epsilon/4)\mu} \le 2 e^{-\mu \epsilon^2/64}= o(1).
\end{align*}

Second, we show that $N_j$ concentrates around $\pi_j \mu$ for each $j \in  Z$,
\[
\mathbb{P}\Big(|N_j - \pi_j \mu| \geq \frac{\epsilon}{4} \pi_j \mu\Big) \leq 2\exp\left(-\frac{\epsilon^2 \pi_j \mu}{64}\right)= o(1).
\]
A union bound over all $j \in  Z$ gives $N_j \in ((1 - \epsilon/4)\pi_j\mu,\; (1 + \epsilon/4)\pi_j\mu)$ for all $j$, with high probability.

Using the intersection of the two events above, we can upper bound the difference with high probability. For each $j \in  Z$
\[
\frac{N_j}{|V_{\mathrm{init}}|} \leq \frac{(1+\epsilon/4)\pi_j\mu}{(1-\epsilon/4)\mu} = \pi_j\cdot\frac{1+\epsilon/4}{1-\epsilon/4} \leq \pi_j + \epsilon,
\]
where the last inequality holds since $\epsilon < \min_{j \in Z}\pi_{j}/2$. The argument for the lower bound $\frac{N_j}{|V_{\mathrm{init}}|} \geq \pi_j - \epsilon$ follows symmetrically. Hence $\sigma^* \in X_0^*(\epsilon)$ with high probability.
\end{proof}

We aim to show that the log-likelihood of any incorrect labeling $\sigma: {V_{\mathrm{init}}} \rightarrow {Z}$ is strictly smaller than that of the true labeling $\sigma^*$.

We distinguish between labelings that are close to the ground truth and those that are significantly different. We formalize this using the Hamming distance up to permissible relabelings, which we term as \textit{discrepancy}.

\begin{definition}[Discrepancy]
The Hamming distance between two labelings $\sigma, \sigma': {V_{\mathrm{init}}} \to {Z}$ is defined by
\begin{equation}
    d_H(\sigma, \sigma') = \sum_{u \in {V_{\mathrm{init}}}} \1(\sigma(u) \neq \sigma'(u)).
\end{equation}
The discrepancy between $\sigma$ and $\sigma'$ is defined as the minimum Hamming distance achieved over the set of permissible relabelings $\Omega_{\pi,P}$:
\begin{equation}
    \text{DISC}(\sigma, \sigma') = \min_{\omega \in \Omega_{\pi,P}} d_H(\omega \circ \sigma, \sigma').
\end{equation}
\end{definition}

A labeling $\sigma$ is considered correct if and only if $\text{DISC}(\sigma, \sigma^*) = 0$. We establish the optimality of the true labeling by showing that the log-likelihood of any incorrect labeling is strictly lower than that of the ground truth.
\begin{proposition}[Low Discrepancy]
\label{prop: ld}
Suppose that Assumptions~\ref{ass:identifiability} and \ref{ass:bounded-ll} hold. There exists a constant $c \in (0,1)$ such that with high probability:
\[
\forall \sigma: {V_{\mathrm{init}}} \to 
{Z} \text{ such that } 0 < \mathrm{DISC}(\sigma, \sigma^*) < c \log n, \text{ we have } \ell_0({G}, \sigma) < \ell_0({G}, \sigma^*).
\]
\end{proposition}

\begin{proposition}[High Discrepancy]
\label{prop: hd}
Suppose that Assumption~\ref{ass:init-identifiability} holds. 
Fix any constant $c \in (0,1)$. Then with high probability:
\[
\forall \sigma \in X_0^*(\epsilon)\text{ such that } \mathrm{DISC}(\sigma, \sigma^*) \geq c \log n, \text{ we have } \ell_0({G}, \sigma) < \ell_0({G}, \sigma^*).
\]
\end{proposition}

Throughout this section, we fix $t=1/2$. Recall the definition of $\ell$ and $\pi_{\mathrm{wit}}$ from \eqref{def:ell}. For any $a \neq b$ we fix some interval $I_{ab} \subset [0,r]$ with $|I_{ab}| = \ell $ such that $\sum_{w \in W_{ab}(I_{ab})}\pi_w \ge \pi_{\mathrm{wit}}$ and  for any $w \in W_{ab}(I_{ab})$ and almost all $y \in I_{ab}$,
$$
\phi_{1/2}(P_{aw}(y),P_{bw}(y))<1.
$$
The existence of $I_{ab}$ is guaranteed by Assumption. \ref{ass:identifiability}.

By Lemma \ref{lem: informative interval}, for any $\epsilon>0$ there exists some $\Phi <1$ such that for any $w \in W_{ab}(I_{ab})$,
\[
\Big|\Big\{y \in I_{ab}: \phi_{1/2}(P_{aw}(y),P_{bw}(y)) < \Phi\Big\}\Big| \ge (1-\epsilon) |I_{ab}|=(1-\epsilon) \ell.
\]
For $u \in B_{\mathrm{init}}$ let
\begin{align}
A_{ab}(u):=\{x \in B_{\mathrm{init}}:\,\frac{\|x-u\|}{(\log n)^{1/d}} \in I_{ab}\}.\label{def:Aabu}
\end{align}
By Lemma~\ref{lem: informative interval} and the geometry of
$B_{\mathrm{init}}$, there exist constants
$\Phi_1<1$ and $\kappa_0>0$ such that, uniformly over
$a\neq b$, $w\in W_{ab}(I_{ab})$, and
$u\in B_{\mathrm{init}}$,
\begin{align}
\left|
\left\{
x\in A_{ab}(u):
\phi_{1/2}\left(
\bar P_{aw}(\|x-u\|),
\bar P_{bw}(\|x-u\|)
\right)
<\Phi_1
\right\}
\right|
\geq
\kappa_0\log n.
\label{eq:Aab-informative-volume}
\end{align}

Let
\[
\mathfrak R
:=
\left\{
(\omega,a,b):
\begin{array}{l}
\omega:Z\to Z\text{ is a permutation},\\
\pi_{\omega(i)}=\pi_i\ \text{for all }i\in Z,\\
P_{ab}\not\equiv P_{\omega(a)\omega(b)}
\end{array}
\right\}.
\]
For every $(\omega,a,b)\in\mathfrak R$, there exists
$\Phi_{ab}^{\omega}<1$ such that
\[
\left|
\left\{
y\in[0,r]:
\phi_{1/2}\left(
P_{\omega(a)\omega(b)}(y),P_{ab}(y)
\right)
\leq\Phi_{ab}^{\omega}
\right\}
\right|>0.
\]
Since $\mathfrak R$ is finite, we may set
\[
\Phi_2
:=
\max_{(\omega,a,b)\in\mathfrak R}
\Phi_{ab}^{\omega}<1.
\]
For $(\omega,a,b)\in\mathfrak R$, let
\[
J_{ab}^{\omega}
:=
\left\{
y\in[0,r]:
\phi_{1/2}\left(
P_{\omega(a)\omega(b)}(y),
P_{ab}(y)
\right)
\leq\Phi_2
\right\}.
\]
Then $|J_{ab}^{\omega}|>0$ for every
$(\omega,a,b)\in\mathfrak R$.
For measurable $J \subseteq [0,r]$ let 
\[
A_J(u)
:=
\left\{
x\in B_{\mathrm{init}}:
\frac{\|x-u\|}{(\log n)^{1/d}}\in J
\right\}.
\]

There is a determininistic constant $\kappa>0$, uniform over all $(\omega,a,b) \in \mathfrak R$ and all $u \in B_{\mathrm{init}}$, such that
\[
|A_{J_{ab}^{\omega}}|\ge \kappa \log n.
\]
Indeed, we may choose
\[
\kappa=\frac{d \nu_d}{2^d} \min_{(\omega,a,b) \in \mathfrak R} \int_{J_{ab}^\omega} s^{d-1} \mathrm d s>0.
\]
Now let $m_0:=\lambda \pi_{\mathrm{min}} \kappa$, $\eta:=m_0/2>0$ and
\begin{align}
\mathcal G_{\mathrm{shell}}
:=
\bigcap_{(\omega,a,b)\in\mathfrak R}
\bigcap_{u\in V_{\mathrm{init}}}
\left\{
\left|
(V_b\setminus\{u\})
\cap A_{J_{ab}^{\omega}}(u)
\right|
\geq
\eta\log n
\right\}.
\label{eq:G-shell}
\end{align}
Then,
\begin{align*}
\mathbb P(\mathcal G_{\mathrm{shell}}^c)
&\leq
\mathbb E\left[
\sum_{u\in V_{\mathrm{init}}}
\sum_{(\omega,a,b)\in\mathfrak R}
\mathbf 1\left\{
|(V_b\setminus\{u\})
\cap A_{J_{ab}^{\omega}}(u)|
<
\eta\log n
\right\}
\right]
\\
&\leq
\lambda |B_{\mathrm{init}}|
|\mathfrak R|\,
n^{-m_0/8}
=o(1),
\end{align*}
where we have used the Mecke equation (Proposition \ref{thm:mecke}) to obtain the second inequality.

Fix $c_-,c_+>0$ such that the event
\[
\mathcal G_0:=\{c_-\log n \le|V_{\mathrm{init}}| \le c_+ \log n\}
\]
occurs with high probability. Fix $c,\epsilon>0$ such that 
\begin{equation}
\label{eq: c and eps}
c\leq \frac{\lambda\pi_{\mathrm{wit}}\kappa_0
\log(1/\Phi_1)}{4\rho},\quad \epsilon \leq \min \Big\{\frac{ \pi_{\min}}{2}, \underset{i,j \in {Z}, \pi_i \neq \pi_j}{\min}\frac{|\pi_i - \pi_j|}{3}, \frac{c}{(k-1)c_+},\frac{k\eta}{2(k-1)c_+}\Big\}.
\end{equation}

\begin{proof}[Proof of Proposition \ref{prop: initial ball}] 
Consider the estimated labeling $\hat{\sigma}$ that maximizes the restricted MLE, i.e.
\begin{equation}
    \hat{\sigma} \in \underset{\sigma \in X^*_0(\epsilon)}{\arg \max } \ \ell_0(G,\sigma).
\end{equation}
The probability that $\hat \sigma$ does not achieve exact recovery (i.e.~that $\hat \sigma$ does not coincide with $\sigma ^*$) can be bounded by
\begin{equation}
\begin{split}
 &\mathbb{P}\Big( \bigcup_{\substack{\sigma\in X_0^*(\epsilon)\\ \text{DISC}(\sigma,\sigma^*) \neq 0}} \{\ell_0(G,\sigma) \ge \ell_0(G,\sigma^*)\} \Big)\\
 \quad &\le 
 \mathbb{P}\Big( \bigcup_{\substack{\sigma \in X_0^*(\epsilon)\\ \text{DISC}(\sigma,\sigma^*) \neq 0}} \{\ell_0(G,\sigma) \ge \ell_0(G,\sigma^*)\} \cap\{\sigma^* \in X_0^*(\epsilon)\} \Big) + \mathbb P(\sigma^* \notin X_0^*(\epsilon)). 
 \end{split}
\end{equation}
Proposition \ref{prop: ld} and \ref{prop: hd} yield that the first probability tends to zero, and Lemma \ref{lem:restr mle proof} gives that the second probability tends to zero. This finishes the proof.
\end{proof}
Now it remains to prove Propositions \ref{prop: ld} and \ref{prop: hd}.


\begin{proof}[Proof of Proposition \ref{prop: ld}]
Fix some $\sigma: {V_{\mathrm{init}}} \to {Z}$ with $\text{DISC}(\sigma, \sigma^*) > 0$. Without loss of generality, after a permissible relabeling, we assume that $d_H(\sigma, \sigma^*) = \text{DISC}(\sigma, \sigma^*)$.

The log-likelihood difference between $\sigma$ and the ground truth labeling $\sigma^*$ can be expressed as  

\begin{align*}
&\ell_0(G, \sigma) - \ell_0(G, \sigma^*) = \sum_{u \in {V_{\mathrm{init}}}} \sum_{\substack{v \in {V_{\mathrm{init}}} \\ v \neq u}} \log \left( \frac{\bar{p}_{\sigma(u) \sigma(v)}(x_{uv}; \|u - v\|)}{\bar{p}_{\sigma^*(u) \sigma^*(v)}(x_{uv}; \|u - v\|)} \right) \\
&= \sum_{u \in {V_{\mathrm{init}}}} \sum_{\substack{v \in {V_{\mathrm{init}}} \\ v \neq u}} \log \left( \frac{\bar{p}_{\sigma(u) \sigma^*(v)}(x_{uv}; \|u - v\|)}{\bar{p}_{\sigma^*(u) \sigma^*(v)}(x_{uv}; \|u - v\|)} \right)  + \sum_{u \in {V_{\mathrm{init}}}} \sum_{\substack{v \in {V_{\mathrm{init}}} \\ v \neq u}} \log \left( \frac{\bar{p}_{\sigma(u) \sigma(v)}(x_{uv}; \|u - v\|)}{\bar{p}_{\sigma(u) \sigma^*(v)}(x_{uv}; \|u - v\|)} \right) \\
&= \sum_{\substack{u \in {V_{\mathrm{init}}} \\ \sigma(u) \neq \sigma^*(u)}} \sum_{\substack{v \in {V_{\mathrm{init}}} \\ v \neq u}} \log \left( \frac{\bar{p}_{\sigma(u) \sigma^*(v)}(x_{uv}; \|u - v\|)}{\bar{p}_{\sigma^*(u) \sigma^*(v)}(x_{uv}; \|u - v\|)} \right)  + \sum_{\substack{v \in {V_{\mathrm{init}}} \\ \sigma(v) \neq \sigma^*(v)}} \sum_{\substack{u \in {V_{\mathrm{init}}} \\ u \neq v}} \log \left( \frac{\bar{p}_{\sigma(u) \sigma(v)}(x_{uv}; \|u - v\|)}{\bar{p}_{\sigma(u) \sigma^*(v)}(x_{uv}; \|u - v\|)} \right).
\end{align*}

We decompose the last expression into $A(\sigma) + B(\sigma)$ with
\begin{align*}
A(\sigma) := \sum_{\substack{u \in {V_{\mathrm{init}}} \\ \sigma(u) \neq \sigma^*(u)}} \sum_{\substack{v \in {V_{\mathrm{init}}} \\ v \neq u}} \log \left( \frac{\bar{p}_{\sigma(u) \sigma^*(v)}(x_{uv}; \|u - v\|)}{\bar{p}_{\sigma^*(u) \sigma^*(v)}(x_{uv}; \|u - v\|)} \right).
\end{align*}
Since there are exactly $d_H(\sigma, \sigma^*)$ points with $\sigma(u) \neq \sigma^*(u)$, it holds that
\begin{align*}
&\mathbb{P}\Big( \bigcup_{\sigma: {V_{\mathrm{init}}} \to {Z}} \{ A(\sigma) > -d_H(\sigma, \sigma^*) c_1 \log n \} \Big)\\
&\quad \le \mathbb{P}\Big( \bigcup_{u \in {V_{\mathrm{init}}}} \bigcup_{b \neq \sigma^*(u)} \Big\{ \sum_{\substack{v \in {V_{\mathrm{init}}} \\ v \neq u}} \log \Big( \frac{\bar{p}_{b \sigma^*(v)}(x_{uv}; \|u - v\|)}{\bar{p}_{\sigma^*(u) \sigma^*(v)}(x_{uv}; \|u - v\|)} \Big) > -c_1 \log n \Big\} \Big)\\
&\quad \le \E\Big[\sum_{u \in {V_{\mathrm{init}}}} \sum_{b \neq \sigma^*(u)} \1\Big\{ \sum_{\substack{v \in {V_{\mathrm{init}}} \\ v \neq u}} \log \Big( \frac{\bar{p}_{b \sigma^*(v)}(x_{uv}; \|u - v\|)}{\bar{p}_{\sigma^*(u) \sigma^*(v)}(x_{uv}; \|u - v\|)} \Big) > -c_1 \log n \Big\} \Big].
\end{align*}
By the Mecke equation (Proposition \ref{thm:mecke}), the above is given by
\begin{align}
  &\lambda  \int_{B_{\mathrm{init}}} \sum_{a \in Z} \pi_a \sum_{b \in Z\setminus \{a\}}\P\Big(\sum_{v \in {V_{\mathrm{init}}}} \log \Big( \frac{\bar{p}_{b \sigma^*(v)}(x_{uv}; \|u - v\|)}{\bar{p}_{a \sigma^*(v)}(x_{uv}; \|u - v\|)} \Big) > -c_1 \log n \Big)\mathrm{d}u.\label{eq:Asigma_mecke}
\end{align}
Fix $a,b \in Z$ with $a \neq b$ and let
$$
S_u^{ab} = \sum_{v \in V_{\mathrm{init}}} \log \Big( \frac{\bar{p}_{b \sigma^*(v)}(x_{uv}; \|u - v\|)}{\bar{p}_{a \sigma^*(v)}(x_{uv}; \|u - v\|)} \Big)
 $$ be the log-likelihood contribution of $u$. 
 
For fixed $a \neq b$, define a set of informative marked points
\[
E_{ab}(u)\coloneqq \{(x,w):x \in A_{ab}(u), w\in W_{ab}(I_{ab}),\phi_t(\bar P_{aw}(\|u-x\|),\bar P_{bw}(\|u-x\|))<\Phi_1\},
\]
where $A_{ab}(u)$ is defined at $\eqref{def:Aabu}$.
By \eqref{eq:Aab-informative-volume}, uniformly in $u$,
\[
\sum_{w \in W_{ab}(I_{ab})} \pi_w |\{x:(x,w) \in E_{ab}(u)\}| \ge \kappa_0 \pi_{\mathrm{wit}} \log n.\]
Therefore the number $N_u^{ab}$ of Poisson points falling in $E_{ab}(u)$ satisfies
\[
N_u^{ab} \sim \mathrm{Poisson}(\mu_u^{ab}),\qquad \mu_u^{ab} \ge \kappa_0 \lambda \pi_{\mathrm{wit}} \log n.
\]

Hence, by a Poisson Chernoff bound (Lemma \ref{lemma:poisson_concentration}),
\begin{equation}
\mathbb{P}\left(N_u^{ab} \le \frac{2}{3}\mu_u^{ab}\right) \le \exp\left(-\frac{1}{18} \mu_u^{ab}\right) \le n^{-\lambda \pi_{\mathrm{wit}}\kappa_0/18}.\label{eq:Nubound}
\end{equation}
Let
\[
c_1
:=
\frac14
\lambda\pi_{\mathrm{wit}}\kappa_0
\log(1/\Phi_1)>0.
\]
We bound the probability that $S_u^{ab}$ exceeds
$-c_1\log n$ by distinguishing according to the value of
$N_u^{ab}$. Since $t=1/2$, Markov's inequality gives
\begin{align}
\mathbb P(S_u^{ab}\ge -c_1\log n)
&\le
\mathbb P\left(
\{S_u^{ab}\ge -c_1\log n\}
\cap
\left\{
N_u^{ab}>\frac23\mu_u^{ab}
\right\}
\right)
\nonumber\\
&\qquad+
\mathbb P\left(
N_u^{ab}\le\frac23\mu_u^{ab}
\right)
\nonumber\\
&\le
n^{c_1/2}
\mathbb E\left[
e^{S_u^{ab}/2}
\mathbf 1\left\{
N_u^{ab}>\frac23\mu_u^{ab}
\right\}
\right]
+
\mathbb P\left(
N_u^{ab}\le\frac23\mu_u^{ab}
\right).
\label{eq:Su_bound}
\end{align}
Conditioning on $V$ and $\sigma^*$, we have
\[
\mathbb E\left[
e^{S_u^{ab}/2}
\,\middle|\,V,\sigma^*
\right]
=
\prod_{v\in V_{\mathrm{init}}}
\phi_{1/2}\left(
\bar P_{b\sigma^*(v)}(\|u-v\|),
\bar P_{a\sigma^*(v)}(\|u-v\|)
\right)
\le
\Phi_1^{N_u^{ab}}.
\]
Consequently,
\begin{align*}
\mathbb P(S_u^{ab}\ge-c_1\log n)
&\le
n^{c_1/2}
\Phi_1^{\frac23\mu_u^{ab}}
+
n^{-\lambda\pi_{\mathrm{wit}}\kappa_0/18}
\\
&\le
n^{c_1/2}
n^{-\frac23\lambda\pi_{\mathrm{wit}}\kappa_0
\log(1/\Phi_1)}
+
n^{-\lambda\pi_{\mathrm{wit}}\kappa_0/18}
\\
&=
n^{-\frac{13}{24}
\lambda\pi_{\mathrm{wit}}\kappa_0
\log(1/\Phi_1)}
+
n^{-\lambda\pi_{\mathrm{wit}}\kappa_0/18}
\\
&=
o(1/\log n).
\end{align*}

Together with \eqref{eq:Asigma_mecke}, this gives
\begin{equation}
\mathbb{P}\left( \bigcap_{\sigma: {V_{\mathrm{init}}} \to {Z}} \{ A(\sigma) \le -d_H(\sigma, \sigma^*) c_1 \log n \} \right) = 1 - o(1).
\end{equation}

It remains to bound $B(\sigma)$. Following the decomposition of the term $B(\sigma)$ in the proof of
\cite{gaudio2026jan}, we split
\begin{equation}
B(\sigma) = B_1(\sigma) + B_2(\sigma),
\end{equation}
where
\begin{equation}
B_1(\sigma) = \sum_{\substack{v \in {V_{\mathrm{init}}} \\ \sigma(v) \neq \sigma^*(v)}} 
\sum_{\substack{u \in {V_{\mathrm{init}}} \\ u \neq v}} 
\log \frac{\bar{p}_{\sigma^*(u), \sigma(v)}(x_{uv}; \|u - v\|)}
{\bar{p}_{\sigma^*(u), \sigma^*(v)}(x_{uv}; \|u - v\|)}
\end{equation}
which is identical to $A(\sigma)$ and
\begin{equation}
B_2(\sigma) = \sum_{\substack{u, v \in {V_{\mathrm{init}}}, \; u \neq v \\ 
\sigma(u) \neq \sigma^*(u), \, \sigma(v) \neq \sigma^*(v)}} 
\log \frac{\bar{p}_{\sigma(u), \sigma(v)}(x_{uv}; \|u - v\|) \, 
\bar{p}_{\sigma^*(u), \sigma^*(v)}(x_{uv}; \|u - v\|)}
{\bar{p}_{\sigma(u), \sigma^*(v)}(x_{uv}; \|u - v\|) \, 
\bar{p}_{\sigma^*(u), \sigma(v)}(x_{uv}; \|u - v\|)}.
\end{equation}
By Assumption \ref{ass:bounded-ll}, each summand of $B_2(\sigma)$ is bounded by 
$2\rho$, the number of summands is at most $d_H(\sigma, \sigma^*)^2$. Hence 
$B_2(\sigma) \le 2\rho\, d_H(\sigma, \sigma^*)^2$.
Combining these bounds, with high probability
\begin{equation}
\ell_0(G, \sigma) - \ell_0(G, \sigma^*) 
\;\leq\; -2\, d_H(\sigma, \sigma^*)\, c_1 \log n + 2\rho\, d_H(\sigma, \sigma^*)^2
\end{equation}
for every $\sigma$. When $0 < d_H(\sigma, \sigma^*) < c \log n$ with $c <c_1/\rho$, 
the right-hand side is strictly negative. This concludes the proof of exact 
recovery in the initial cube in the low-discrepancy case.
Thus  $\ell_0(G, \sigma) - \ell_0(G, \sigma^*) < 0$, with high probability, when  $0 <\text{DISC}(\sigma, \sigma^*) <c \log n $ with $c <\frac{c_1}{\rho}$. 
\end{proof}

\begin{lemma}[Geometric split]
\label{lem:geometric-split}
Fix $s\in(0,r]$ and $\delta>0$. There exists
$\gamma=\gamma(s,\delta)>0$ such that, with probability $1-o(1)$,
the following holds simultaneously for every $a\in Z$ and every
labeling
\[
\sigma:V\cap B_{\mathrm{init}}\to Z.
\]

Suppose that there exist $b,b'\in Z$, with $b\neq b'$, such that
\[
\left|
\{u\in V_a\cap B_{\mathrm{init}}:\sigma(u)=b\}
\right|
\geq\delta\log n
\]
and
\[
\left|
\{u\in V_a\cap B_{\mathrm{init}}:\sigma(u)=b'\}
\right|
\geq\delta\log n.
\]
Then there exists a collection $\mathcal E$ of unordered pairs
$\{u,v\}\subseteq V_a\cap B_{\mathrm{init}}$ such that
\[
|\mathcal E|
\geq
\gamma(\log n)^2,
\]
and, for every $\{u,v\}\in\mathcal E$,
\[
\sigma(u)\neq\sigma(v)
\qquad\text{and}\qquad
\|u-v\|
\leq
s(\log n)^{1/d}.
\]
\end{lemma}

\begin{proof}
Rescale $B_{\mathrm{init}}$ by $(\log n)^{-1/d}$, so that it
becomes the fixed cube $[-r,r]^d$.

Choose an integer $m\geq1$ sufficiently large that, if
$[-r,r]^d$ is partitioned into $m^d$ congruent cubes, then any two
points lying either in the same cube or in two face-adjacent cubes
are at Euclidean distance at most $s$.

Let
\[
\mathcal C=\{C_1,\ldots,C_M\},
\qquad M=m^d,
\]
denote the corresponding partition of $B_{\mathrm{init}}$, scaled
back by $(\log n)^{1/d}$.

Since the points of community $a$ form a Poisson point process of
intensity $\lambda\pi_a$, there exists $\alpha>0$ such that, with
probability $1-o(1)$,
\begin{equation}
\label{eq:true-community-cell-occupancy}
|V_a\cap C|
\geq
\alpha\log n
\end{equation}
simultaneously for every $a\in Z$ and $C\in\mathcal C$.
Indeed, there are only $kM=O(1)$ such random variables, and each is
Poisson with mean of order $\log n$.

Work on the event \eqref{eq:true-community-cell-occupancy}. Set
\[
\tau
:=
\min\left\{
\frac{\delta}{2M},
\frac{\alpha}{2k}
\right\}.
\]
For $C\in\mathcal C$ and $j\in Z$, call $j$ \emph{heavy in $C$} if
\[
\left|
\{u\in V_a\cap C:\sigma(u)=j\}
\right|
\geq
\tau\log n.
\]

By \eqref{eq:true-community-cell-occupancy} and the pigeonhole
principle, every cell has at least one heavy label. Moreover, since
\[
\left|
\{u\in V_a\cap B_{\mathrm{init}}:\sigma(u)=b\}
\right|
\geq\delta\log n,
\]
there exists a cell in which $b$ is heavy. The same holds for $b'$.

Suppose first that some cell $C$ contains two distinct heavy labels
$j\neq j'$. Then there are at least $tau^2(\log n)^2$ pairs $\{u,v\}\subseteq V_a\cap C$ with $\sigma(u)=j$and $sigma(v)=j'$.
By the choice of the partition,
\[
\|u-v\|
\leq
s(\log n)^{1/d},
\]
and the result follows.

It remains to consider the case in which every cell has a unique
heavy label. Assign to each cell its unique heavy label. The
face-adjacency graph of the cells is connected. Since some cell has
heavy label $b$ and some other cell has heavy label $b'$, there must
exist two face-adjacent cells $C,C'$ whose heavy labels, say $j$ and
$j'$, are distinct.

Consequently, there are at least $\tau^2(\log n)^2$ pairs
$u\in V_a\cap C$, $v\in V_a\cap C'$ such that
\[
\sigma(u)=j,
\qquad
\sigma(v)=j'.
\]
Again, by construction,
\[
\|u-v\|
\leq
s(\log n)^{1/d}.
\]
Thus, the assertion holds with $\gamma:=\tau^2$.
\end{proof}

\begin{proof}[Proof of Proposition \ref{prop: hd}] 
Let
\[
\mathcal G_0
:=
\left\{
\sigma^*\in X_0^*(\epsilon),\
c_-\log n
\leq |V_{\mathrm{init}}|
\leq c_+\log n
\right\},
\]
and let $\mathcal G_{\mathrm{split}}$ be the high-probability event given by Lemma \ref{lem:geometric-split} for
\[
s=r_{\mathrm{init}},\quad \delta=\frac{\epsilon c_-}{k}.
\]
We work on the event $\mathcal G:=\mathcal G_0 \cap \mathcal G_{\mathrm{split}} \cap \mathcal G_{\mathrm{shell}}$. Then $\P(\mathcal G)=1-o(1)$.
Recall the definitions of $c>0$ and let $\epsilon>0$ from \eqref{eq: c and eps}. We consider a labeling $\sigma\in X_0^*(\epsilon)$ satisfying
\[
\mathrm{DISC}(\sigma,\sigma^*)\geq c\log n.
\]
For $i,j\in Z$, let
\[
N_{ij}
=
\{u\in V_{\mathrm{init}}:
\sigma^*(u)=i,\ \sigma(u)=j\},
\qquad
n_{ij}=|N_{ij}|.
\]

Since $P_{ij}=P_{ji}$ for all $i,j \in Z$, we have, 
\[
\ell_0(G,\sigma)-\ell_0(G,\sigma^*) = 2 \sum_{\{u,v\} \subseteq V_{\mathrm{init}}}
 \log \left( \frac{\bar{p}_{\sigma(u), \sigma(v)}(x_{uv}; \|u - v\|)}{\bar{p}_{\sigma^*(u), \sigma^*(v)}(x_{uv}; \|u - v\|)} \right)
\]
Conditionally on $V$ and $\sigma^*$, Markov's inequality gives
\begin{align}
&\mathbb P\left(
\ell_0(G,\sigma)-\ell_0(G,\sigma^*)\geq0
\,\middle|\,V,\sigma^*
\right)
\nonumber\\
&\qquad\leq
\mathbb E\left[
\exp\left\{
\frac14\bigl(
\ell_0(G,\sigma)-\ell_0(G,\sigma^*)
\bigr)
\right\}
\,\middle|\,V,\sigma^*
\right]
\nonumber\\
&\qquad=
\prod_{\{u,v\}\subseteq V_{\mathrm{init}}}
\phi_{1/2}\left(
\overline P_{\sigma(u)\sigma(v)}(\|u-v\|),
\overline P_{\sigma^*(u)\sigma^*(v)}(\|u-v\|)
\right).
\label{eq:l0diff_high}
\end{align}

Since the vertex locations have an absolutely continuous
distribution, almost surely none of the finitely many interpoint
distances in $V_{\mathrm{init}}$ belongs to any exceptional
null set in Assumption~\ref{ass:init-identifiability}.
We now distinguish two cases.
\medskip

\noindent
\textbf{Case 1.}
Suppose that there exist $a,b,b'\in Z$, with $b\neq b'$, such that
\[
n_{ab},n_{ab'}
\geq
\frac{\epsilon}{k}|V_{\mathrm{init}}|.
\]
On the event $|V_{\mathrm{init}}|\geq c_-\log n$, Lemma
\ref{lem:geometric-split}, applied with
\[
s=r_{\mathrm{init}},
\qquad
\delta=\frac{\epsilon c_-}{k},
\]
gives a set $\mathcal E$ of at least
$\gamma(\log n)^2$ unordered pairs $u,v \in V_{\mathrm{init}}$ such that
\[
\sigma^*(u)=\sigma^*(v),
\qquad
\sigma(u)\neq\sigma(v),
\qquad
\|u-v\|
\leq
r_{\mathrm{init}}(\log n)^{1/d}
\]
for every $\{u,v\}\in\mathcal E$.

Assumption~\ref{ass:init-identifiability} therefore yields
\[
\phi_{1/2}\left(
\bar P_{\sigma(u)\sigma(v)}(\|u-v\|),
\bar P_{\sigma^*(u)\sigma^*(v)}(\|u-v\|)
\right)
\leq
\Phi_{\mathrm{init}}
\]
for every $\{u,v\}\in\mathcal E$. Hence
\[
\mathbb P\left(
\ell_0(G,\sigma)-\ell_0(G,\sigma^*)\geq0
\right)
\leq
\Phi_{\mathrm{init}}^{\gamma(\log n)^2}
=
C_1^{(\log n)^2}
\]
for some $C_1\in(0,1)$.
\medskip

\noindent
\textbf{Case 2.}
Suppose that for every $a\in Z$ there is at most one $b\in Z$ such that
\[
n_{ab}\geq
\frac{\epsilon}{k}|V_{\mathrm{init}}|.
\]
We claim that for every $a$ there is in fact exactly one such $b$.
Indeed, if no such $b$ existed for some $a$, then
\[
\sum_{b\in Z}n_{ab}
<
\epsilon|V_{\mathrm{init}}|.
\]
On the other hand, since $\sigma^*\in X_0^*(\epsilon)$,
\[
\sum_{b\in Z}n_{ab}
=
|\{u:\sigma^*(u)=a\}|
\geq
(\pi_a-\epsilon)|V_{\mathrm{init}}|.
\]
Since $\epsilon\leq\pi_{\min}/2$, this is at least
$\epsilon|V_{\mathrm{init}}|$, which is a contradiction.

Hence there is a unique map $\omega:Z\to Z$ such that
\[
n_{a\omega(a)}
\geq
\frac{\epsilon}{k}|V_{\mathrm{init}}|
\qquad\text{for every }a\in Z.
\]

We next show that $\omega$ is a permutation. Suppose that some
$b\in Z$ is not in the image of $\omega$. Then, by the assumption of
Case 2,
\[
n_{ab}
<
\frac{\epsilon}{k}|V_{\mathrm{init}}|
\qquad\text{for every }a\in Z.
\]
Consequently,
\[
|\{u \in V_{\mathrm{init}}:\sigma(u)=b\}|
=
\sum_{a\in Z} n_{ab}
<
\epsilon|V_{\mathrm{init}}|.
\]
But $\sigma\in X_0^*(\epsilon)$ implies
\[
|\{u:\sigma(u)=b\}|
\geq
(\pi_b-\epsilon)|V_{\mathrm{init}}|
\geq
(\pi_b-\epsilon)|V_{\mathrm{init}}|,
\]
which is at least
$\epsilon|V_{\mathrm{init}}|$ since
$\epsilon\leq\pi_{\min}/2$. This is a contradiction. Thus $\omega$ is
bijective.

We now show that $\omega$ cannot be a permissible relabeling. If it
were permissible, then $\omega^{-1}$ would also be permissible and
\[
\mathrm{DISC}(\sigma,\sigma^*)
\leq
d_H(\omega^{-1}\circ\sigma,\sigma^*)
=
\sum_{a\in Z}\sum_{b\neq\omega(a)}n_{ab}.
\]
By the assumption of Case 2,
\[
\sum_{a\in Z}\sum_{b\neq\omega(a)}n_{ab}
<
k(k-1)
\frac{\epsilon}{k}|V_{\mathrm{init}}|
=
(k-1)\epsilon|V_{\mathrm{init}}|.
\]
Using \eqref{eq: c and eps} gives
\[
\mathrm{DISC}(\sigma,\sigma^*)
<
c\log n,
\]
contradicting the assumption on $\sigma$. Hence $\omega$ is not
permissible.

Following the argument in the proof of Proposition D.6 \cite{gaudio2025incguan}, we next establish that $\pi_a=\pi_{\omega(a)}$ for every $a\in Z$. Suppose that $\pi_{\omega(a)}\neq \pi_a$ for some $a \in Z$. Since $\omega$ is a permutation, there must then exist some $a$ such that $\pi_{\omega(a)}<\pi_a$.  Since $\sum_{b \in Z} n_{ab}\ge (\pi_a - \epsilon) |V_{\mathrm{init}}|$ with high probability, and $n_{ab}<\frac{\epsilon}{k} |V_{\mathrm{init}}|$ for any $b \neq \omega(a)$, we have
\[
n_{a \omega(a)} > (\pi_a-\epsilon) |V_{\mathrm{init}}| - (k-1) |V_{\mathrm{init}}| \frac{\epsilon}{k} > (\pi_a - 2 \epsilon) |V_{\mathrm{init}}|.
\]
Together with \eqref{eq: restr mle} and $|V_{\mathrm{init}}| $, this gives
\[
(\pi_a-2\epsilon) |V_{\mathrm{init}}| < \sum_{b \in Z} n_{b \omega(a)} \le (\pi_{\omega(a)} + \epsilon)  |V_{\mathrm{init}}|.
\]
Rearranging terms yields $\epsilon > (\pi_a - \pi_{\omega(a)})/3$. This contradicts the definition of $\epsilon$. Hence, for all $a \in Z$ we must have $\pi_{\omega(a)} \ge \pi_a$, which implies $\pi_{\omega(a)}=\pi_a$, since the $\{\pi_a\}$ sum to $1$.

Since $\omega$ is not permissible, there exist
$a,b\in Z$ such that
\[
P_{ab}\not\equiv P_{\omega(a)\omega(b)}.
\]
Hence $(\omega,a,b)\in\mathfrak R$. On $\mathcal G_{\mathrm{shell}}$, for every
$u\in N_{a\omega(a)}$,
\[
\left|
(V_b\setminus\{u\})
\cap A_{J_{ab}^{\omega}}(u)
\right|
\geq\eta\log n.
\]

On the other hand, by the assumption of Case~2,
\[
|V_b\setminus N_{b\omega(b)}|
=
\sum_{j\neq\omega(b)}n_{bj}
<
\frac{k-1}{k}\epsilon |V_{\mathrm{init}}|\le \frac{k-1}{k} \epsilon c_+ \log n.
\]
Thus,
\[
|(N_{b\omega(b)}\setminus \{u\})\cap A_{J_{ab}^{\omega}}(u)|
\geq \frac{\eta}{2}\log n.
\]
On $\mathcal G_0$,
\[
n_{a\omega(a)}
\geq\frac{\epsilon}{k}|V_{\mathrm{init}}|
\ge \frac{\epsilon c_-}{k} \log n.
\]
Hence, with $\gamma_2:=(\epsilon c_- \eta)/4k$, there are at least $\gamma_2 (\log n)^2$ unordered pairs $(u,v) \in V_{\mathrm{init}}$ satisfying
\[
\phi_{1/2}\left(
\overline P_{ab}(\|u-v\|),
\overline P_{\omega(a)\omega(b)}(\|u-v\|)
\right)
\leq\Phi_2.
\]
Therefore,
\[
\mathbb P\left(
\ell_0(G,\sigma)\geq\ell_0(G,\sigma^*)
\,\middle|\,V,\sigma^*
\right)
\leq
C^{(\log n)^2}.
\]
Consequently,
\begin{align*}
&\mathbb P\left(
\exists\,\sigma\in X_0^*(\epsilon):
\mathrm{DISC}(\sigma,\sigma^*)\geq c\log n,\
\ell_0(G,\sigma)\geq\ell_0(G,\sigma^*)
\right)
\\
&\qquad\leq
\mathbb P(\mathcal G^c)
+
k^{c_+\log n}C^{(\log n)^2}
=o(1).
\end{align*}
\end{proof}

\section{Propagation proof}

\subsection{Block structure}
We prepare the proof of Theorem \ref{theorem: achiev} by providing details on the discretization of the space $S_{d,n}$ into blocks. Recall the definitions of $\pi_{\mathrm{wit}}$ and $\ell$ from \eqref{def:ell}.

\begin{lemma}[{\cite[Section C]{gaudio2025incguan}}]
\label{lemma:condition_chi_0}
Assume $d=1$ and $\lambda \pi_\mathrm{wit}\ell>1$. Denote $K_1$ the number of blocks $\ell$-close to the left (or the right)  of length $\ell\chi \log n$.
Then, we have  $\lambda \pi_{\mathrm{wit}}K_1 \chi \ell> 1$ for all $\chi$ satisfying \eqref{eq:condition_chi_0_d1}.

Assume $d\geq 2$ and $\lambda \pi_{\mathrm{wit}} \nu_d \ell^d > 1$. Then there exists $\chi_0 > 0$ such that for all $0 <\chi < \chi_0$, \eqref{eq:condition_chi_0} holds.
   
    Moreover, let $K_d$ the number of blocks that are $\ell$-close to a given block, excluding the block itself. Since the torus is translation invariant and the partition is regular, it depends only on $\chi$ and $d$. For all $\chi$ verifying the condition in \eqref{eq:condition_chi_0}, we have, $\lambda \pi_{\mathrm{wit}} K_d\chi \ell^d > 1$.
\end{lemma}
\begin{proof}
    For $d=1$, $K_1 = \lfloor 1/\chi \rfloor-1$, then $\chi K_1 \geq 1 - 2\chi$ and from \eqref{eq:condition_chi_0_d1}, we obtain the result.
    
   For $d\geq 2$, note that $h(x) \coloneqq \lambda \pi_\mathrm{wit} \ell^d \Big( \nu_d \left(
    1-\tfrac{3\sqrt d}{2}x^{1/d}
    \right)^d- x \Big)$ is decreasing, continuous, and $\lim_{x\downarrow 0}h(x) = \lambda \pi_\mathrm{wit} \ell^d \nu_d > 1$. Hence, there exists $\chi'_0>0$ such that for all $0 < \chi < \chi'_0$, $h(\chi) > 1$. Hence, we take $\chi_0 < \min\left\{\chi'_0, \left(\frac{2}{3\sqrt d}\right)^d\right\}$. We define $\ell_d \coloneqq \ell - \tfrac{3}{2}\sqrt{d} \ell \chi^{1/d}$ the proximity radius reduced by one and a half block diagonals.
   For a block
$B\in\mathcal P_\chi$, let $c(B)$ denote its center and define $C(B) \coloneqq \mathcal B\left(
        c(B),
        \ell_d(\log n)^{1/d}
    \right)$. By second condition in  \eqref{eq:condition_chi_0}, $\ell_d > 0$.  We define the set of blocks $B'$ such that $B$ and $B'$ are $\ell$-close by~$U(B) \coloneqq
    \bigcup_{\substack{B'\sim B \\ B'\neq B}} B'$.
    
Note that $C(B) \subset U(B) \cup B$. Hence,
\[
    K_d \chi \ell^d\log n =
    |U(B)| \geq |C(B)| - |B| \geq \ell^d \left(\nu_d \left(1-\tfrac{3\sqrt d}{2}\chi^{1/d}\right)^d- \chi\right)\log n
.\]
Then from \eqref{eq:condition_chi_0},
$\lambda \pi_\mathrm{wit}K_d \chi \ell^d \geq \lambda \pi_\mathrm{wit} \ell^d \left( \nu_d ( 1-\tfrac{3\sqrt d}{2}\chi^{1/d} )^d -\chi\right)>1$.
\end{proof}

\begin{definition}[Spread block]
\label{def: spread block}
A block $B$ is called $\delta$-spread if
\[
V(B)\geq\delta\log n
\]
and
\[
\min_{a\neq b}
\sum_{w\in W_{ab}(I_{ab})}V_w(B)
\geq
\frac34\delta\pi_{\mathrm{wit}}\log n.
\]
\end{definition}

\begin{definition}[Cluster of blocks, {\cite[Definition C.2]{gaudio2025incguan}}]
    Two blocks are adjacent if they share an edge or a corner. We say that a set of blocks $\mathcal B$ is a cluster if for every $B, B' \in \mathcal B$, there is a path of blocks of the form $(B = B_{j_1}, B_{j_2}, \dots, B_{j_m} = B')$, where $B_{j_k} \in \mathcal B$ for $k \in [m]$ and $B_{j_k}, B_{j_{k+1}}$ are adjacent.
\end{definition}

The following lemma implies that with high probability for every block $B \in \mathcal P_\chi$ there exists a $\delta$-spread block $B'$ such that $B$ and $B'$ are $\ell$-close. It generalizes {\cite[Lemma C.8]{gaudio2025incguan}}.

\begin{lemma} 
\label{lemma:certified_block}
    Suppose for $d=1$, $\lambda \pi_{\mathrm{wit}} \ell^d > 1$ and for $d \geq 2$, $\lambda \pi_{\mathrm{wit}} \nu_d \ell^d > 1$. Let $Y_\delta$ be the size of the largest cluster of blocks that are not $\delta$-spread. Then 
    \[
    \P(Y_\delta < K_d)=1-o(1),
    \]
    where $\delta$ and $K_d$ are given at \eqref{eq:Kr}.
\end{lemma}

\begin{proof}
 Let $\mathcal S_{K_d}$ be the (finite) set of cluster shapes of $K_d$ blocks. We have $\Theta(n/\log n)$ blocks, hence there are $O(n/\log n)$ clusters of $K_d$ blocks to consider. Let $(B_{i_1},\ldots,B_{i_{K_d}})$ be a cluster of $K_d$ block. It suffices to show that the probability that all blocks in the cluster are not $\delta$-spread is $o((\log n)/n)$.
 
If $B$ is not $\delta$-spread, then either
$V(B)<\delta\log n$, or there exist $a\neq b$ such that
\[
\sum_{w\in W_{ab}(I_{ab})}V_w(B)
<
\frac34\delta\pi_{\mathrm{wit}}\log n.
\]
Consequently,
\[
\mathbb P(B\text{ is not }\delta\text{-spread})
\leq
\mathbb P(V(B)<\delta\log n)
+
\sum_{a<b}
\mathbb P\left(
\sum_{w\in W_{ab}(I_{ab})}V_w(B)
<
\frac34\delta\pi_{\mathrm{wit}}\log n
\right).
\]

Hence, by independence and by translation invariance,
\begin{align}
    &\mathbb{P}\Big(\bigcap_{j=1}^{K_d}\{B_{i_j}\text{ is not }\delta\text{-spread}\}\Big)\nonumber\\   &\quad\leq \Big(\mathbb{P}\left(V(B)<\delta\log n\right) + \sum_{a< b} \mathbb{P}\Big(\sum_{w\in W_{ab}(I_{ab})}V_w(B)<\delta\log n\Big) \Big)^{K_d}.\label{eq:non_spread}
\end{align}
Since for all $a \neq b$,
$$
    V(B) \ge  \sum_{w\in W_{ab}(I_{ab})}V_w(B) \quad \text{a.s.},
$$
and $|\{\{a,b\}\in Z^2:\,a \neq b\}|=\binom k2$, \eqref{eq:non_spread} is bounded by
\begin{align*}
&\left(1+\binom{k}{2}\right)^{K_d}
\max_{a\neq b}\, 
\mathbb{P}\Big(
\sum_{w\in W_{ab}(I_{ab}} |V_w(B)|
<\delta\log n
\Big)^{K_d}\\
&\quad \le \left(1+\binom{k}{2}\right)^{K_d}\mathbb{P}\Big(\operatorname{Poisson}(\lambda \pi_{\mathrm{wit}}\chi \ell^d\log n)<\delta\log n\Big)^{K_d}.
\end{align*}
By Lemma \ref{lemma:poisson_concentration} it holds that
$$
\mathbb{P}\left(\operatorname{Poisson}(\lambda \pi_{\mathrm{wit}}\chi \ell^d\log n)<\delta\log n\right)^{K_d} \le n^{-\lambda \pi_{\mathrm{wit}}\ell^d \chi K_d}
\Big(\frac{e \lambda \pi_{\mathrm{wit}}\ell^d \chi}{\delta}\Big)^{\delta K_d\log n}.
$$
Since $\big(\frac{ea}{\delta}\big)^{\delta k}<e^{(ak-1)/2}$ when $ak>1$ and $\delta < \frac{(ak-1)^2}{8ak^2}$, we conclude that with $a:=\lambda \pi_{\mathrm{wit}} \ell^d \chi$ and $k:=K_d$ for all $0<\delta <\delta_0$,
$$
\mathbb{P}\left(\operatorname{Poisson}(\lambda \pi_{\mathrm{wit}}\chi \ell^d\log n)<\delta\log n\right)^{K_d} \le n^{-(1+\lambda \pi_{\mathrm{wit}}\ell^d \chi K_d)/2} \in o(1/n).
$$
Thus,
$$\mathbb{P}(Y_\delta\geq K_d) \leq O(n/\log n) \cdot o(1/n) = o(1).$$
\end{proof}

The following topological statement is \cite[Proposition C.9]{gaudio2025incguan}.
\begin{lemma}\cite[Proposition C.9]{gaudio2025incguan}\label{lemma:geometric_boundary_cluster}
Let $d\geq 2$ and~$\Lambda_m=(\mathbb Z/m\mathbb Z)^d$ be the
discrete torus of blocks.
We assume that each block has exactly $K\coloneqq K_d$ close neighbours.
There exists $m_0$ such that for every $m\geq m_0$, the following holds. Let $O\subset\Lambda_m$ and $H(O)$ be the graph with vertex set $O$ where two vertices are adjacent when the corresponding blocks are $\ell$-close.
If $H(O)$ is disconnected, then $\Lambda_m\setminus O$ contains
a cluster $C \subset \Lambda_m\setminus O$ such that $|C|\geq K$.
\end{lemma}

In the following proposition we consider the subgraph $H_{\mathrm{spread}}$ of $H$, induced by all $\delta$-spread blocks. Combining Lemma~\ref{lemma:certified_block} and Lemma \ref{lemma:geometric_boundary_cluster}, we  obtain that every block in $H$ has a neighboring block in $H_{\mathrm{spread}}$. This property of $H_{\mathrm{spread}}$ will be useful in our propagation proof.

\begin{proposition}
\label{prop:connectivity}
Suppose that, for \(d=1\),
\(\lambda \pi_{\mathrm{wit}}\ell>1\), and, for \(d\geq2\),
\(\lambda\pi_{\mathrm{wit}}\nu_d\ell^d>1\). Let \(V\) be a Poisson point
process with intensity \(\lambda\) restricted to \(S_{d,n}\). For
\(\delta\) given by \eqref{eq:condition_delta}, and \(\chi\) satisfying
\eqref{eq:condition_chi_0_d1} if \(d=1\) and
\eqref{eq:condition_chi_0} if \(d\geq2\), the
\((\chi\ell^d\log n,\delta\log n)\)-proximity-spread graph
\(H_{\mathrm{spread}}\) of \(G\) is connected with high probability.
Moreover, with high probability, every block has at least one neighbor in
\(H_{\mathrm{spread}}\).
\end{proposition}

\begin{proof}
By Lemma~\ref{lemma:certified_block}, with high probability there is no
cluster of at least \(K_d\) non-\(\delta\)-spread blocks. On this event,
if \(d=1\), disconnectedness of \(H_{\mathrm{spread}}\) would imply the
existence of \(K_1\) consecutive non-\(\delta\)-spread blocks, a contradiction.
Hence \(H_{\mathrm{spread}}\) is connected when \(d=1\).

Now suppose \(d\geq 2\). On the same event,
Lemma~\ref{lemma:geometric_boundary_cluster} implies that
\(H_{\mathrm{spread}}\) is connected.

It remains to show that every block has a neighbor in
\(H_{\mathrm{spread}}\). If \(B\) is \(\delta\)-spread, then, since
\(H_{\mathrm{spread}}\) is connected, \(B\) has a neighbor in
\(H_{\mathrm{spread}}\) (for sufficiently large \(n\), the graph contains
more than one vertex).

Now suppose that \(B\) is not \(\delta\)-spread. If no block $\ell$-close to \(B\) is \(\delta\)-spread, then \(B\) together with all blocks
$\ell$-close to \(B\) forms a cluster of non-\(\delta\)-spread blocks
of cardinality at least \(K_d\). This contradicts the event from
Lemma~\ref{lemma:certified_block}. Hence \(B\) has a  $\ell$-close
\(\delta\)-spread block, which is a neighbor of \(B\) in
\(H_{\mathrm{spread}}\).
\end{proof}

\subsection{Almost exact recovery in Phase II}

Let $\omega^* \in \Omega_{\pi,P}$ be the permissible relabeling corresponding to the
estimated labels of the initial block, that is,
\[
\hat{\sigma}(v)=\omega^* \circ \sigma^*(v)
\qquad \text{for all } v\in V_{\mathrm{init}}.
\]
Due to Proposition \ref{prop: initial ball}, $\omega^*$ exists with high probability. Since $\omega^\star$ is permissible, the model, the propagation rule,
and the affinity information are invariant under composition with
$(\omega^\star)^{-1}$. Hence, throughout the proof of Phase~II, we may
without loss of generality assume that
\[
\omega^\star=\mathrm{id}.
\]

The following lemma bounds the probability of mislabeling a vertex $v\in V$ in terms of the number of errors among its labeled neighbors and the affinity information available for $v$.
\begin{lemma}[Propagation error bound]\label{lemma:one_step_error_bound}
 Let $G \sim \mathrm{GHCM}(\lambda,n,r,\pi,P(\cdot),d)$, $u \in V$, $S \subset V\setminus \{u\}$ and $\hat \sigma_S:S\to Z$ be a labeling on $S$. Let $\hat \sigma(u)$ be the output of Propagate (Algorithm \ref{alg:propagate}) applied on input $(G,\{u\},S,\hat \sigma_S)$, assume that Assumption~\ref{ass:bounded-ll} holds with some constant $\rho>0$. Let $\mathcal F$ be a $\sigma$-field with respect to which
$(V,\sigma^*,\widehat\sigma_S)$ is measurable and such that,
conditionally on $(V,\sigma^*)$, the random variables
\[
\{X_{uv}:v\in S\cap N(u)\}
\]
are independent of $\mathcal F$. Then,
    \begin{align*}
        \mathbb P(\hat\sigma(u) \neq \omega^* \circ \sigma^*(u) \mid \mathcal F) \leq (k-1)e^{2\rho m_{\hat \sigma}(S)-\mathcal I_S(u)},
    \end{align*}
    where $m_{\hat \sigma}(S):=|\{v \in S \cap N(u):\;\omega^* \circ \sigma^*(v)\neq \hat \sigma_S(v)\}|$ is the number of incorrectly labeled nodes in $S$ and $\mathcal I_S(u)$ is given at \eqref{eq:def_I}.
\end{lemma}
For Algorithms \ref{alg:exact-recovery} and \ref{alg:phase2_block}, the condition on $\mathcal F$ is satisfied with $F=\sigma(V,\sigma^*,\hat \sigma_S)$, since an edge incident to an unlabeled vertex is
not used before that vertex is labeled.
\begin{proof}
      Let $a:=\omega^* \circ \sigma^*(u)$. Note that, on the event $\{a \neq \hat \sigma(u)\}$, the likelihood of $u$ having label $b \in Z \setminus \{a\}$ exceeds the likelihood of $u$ having label $a$, i.e.
    \begin{align*}
        \Lambda_{b} \coloneqq \sum_{v \in S \cap N(u)} \log \frac{\bar p_{b \hat{\sigma}(v)}(X_{uv};\|u-v\|)}{\bar p_{a \hat{\sigma}(v)}(X_{uv};\|u-v\|)} \geq 0,
    \end{align*}
       where we write $\hat \sigma(v):=\hat\sigma_S(v)$ for $v\in S$. Then, by Markov's inequality, we have for all $t \in (0,1)$, 
    \begin{align*}
        \mathbb{P}(\Lambda_{b} \geq 0 \mid \mathcal F) &= \mathbb{P}(e^{t\Lambda_{b}} \geq 1 \mid \mathcal F) \\
        &\leq \mathbb{E}(e^{t\Lambda_{b}} \mid \mathcal F) \\
        &= \prod_{v \in S \cap N(u)} \int_{\mathbb X} \Big(\frac{\bar{p}_{b \hat{\sigma}(v)}(x;\|u-v\|)}{\bar{p}_{a \hat{\sigma}(v)}(x;\|u-v\|)}\Big)^t P_{a \,\omega^* \circ\sigma^*(v)}(\mathrm d x).
    \end{align*} 
    We distinguish in the product above whether $\hat{\sigma}(v)$ and $\omega^* \circ  \sigma^*(v)$ coincide or not. By Assumption~\ref{ass:bounded-ll}, we have for all $a,b,a',b' \in Z$, all $x \in \mathbb X$ and all $y \in [0,r]$,
    $$
          \bigg|\log \frac{{p}_{ab}(x; y)}{{p}_{a'b'}(x; y)}\bigg| < \rho.
    $$
    Using the notation $\phi_t$ introduced at \eqref{eq:def_phi}, we obtain the bound
    \begin{align*}
       \mathbb{P}(\Lambda_{b} \geq 0 \mid \mathcal F) &\leq  e^{t\rho m_{\hat \sigma}(S)} \cdot \prod_{\substack{v \in S \\ \hat{\sigma}(v) = \omega^* \circ \sigma^*(v)}} \phi_t \Big(\bar{P}_{b\hat{\sigma}(v)}(\|u-v\|), \bar{P}_{a\hat{\sigma}(v)}(\|u-v\|)\Big)\\
       &\leq e^{2\rho m_{\hat \sigma}(S)} \cdot \prod_{v \in S} \phi_t \Big(\bar{P}_{b\hat{\sigma}(v)}(\|u-v\|), \bar{P}_{a\hat{\sigma}(v)}(\|u-v\|)\Big).
    \end{align*}
    Using the definition of $\mathcal I_S(u)$ and a union bound over all $b \in Z\setminus \{a\}$, we arrive at the asserted bound.
\end{proof}

Our next goal is to establish a uniform upper bound on the number of mislabeled nodes that are not labeled $*$ in the neighborhood of every vertex. We will later use this bound to obtain a logarithmic upper bound on the total number of mislabelings in every neighborhood.
\begin{proposition}[Uniform error bound]
\label{prop:abs_propagation}
Assume that Assumptions~\ref{ass:identifiability} and~\ref{ass:bounded-ll} hold with some constant
$\rho>0$ and that Algorithm~\ref{alg:exact-recovery} or \ref{alg:phase2_block} is used. There exists $m\in\mathbb N$ such that
\[
\P\Big(
\bigcap_{w\in V}
\left\{
\left|
\left\{
u\in V\cap N(w):
\widehat\sigma(u)
\notin
\{\omega^\star\circ\sigma^\star(u),*\}
\right\}
\right|
\le m
\right\} \Big) = 1-o(1).
\]
\end{proposition}

The main idea behind the proof of the last proposition is an inductive argument based on the amount of affinity information available for a vertex $w$. More precisely, if we can control the number of errors in the source of $w$ (as is the case, in particular, for vertices close to the initial set) and $w$ has sufficiently much affinity information, then Proposition~\ref{lemma:one_step_error_bound} implies that $w$ is likely to be labeled correctly. It therefore remains to show that controlling the number of errors in the labeled neighborhood of $w$ guarantees the availability of sufficient affinity information. We establish this by showing that, with high probability, there exists a chain of $\ell$-close blocks, each containing sufficiently many vertices of every relevant community. We begin with the following definition.

\begin{definition}[Proximity chain]
A sequence of blocks $\mathcal P=(B_1,\ldots,B_k)$ is called a
proximity chain if, for every $i<k$,
\[
\sup_{x\in B_i, z\in B_{i+1}}\|x-z\|
\le \ell(\log n)^{1/d}.
\]
\end{definition}

The next lemma connects proximity chains to the intervals used to identify witness communities. For $v \in \mathcal S_{d,n}$ and $I \subseteq [0,r]$ recall the definition
\[
\bar{\mathcal A}_I(v):=\Big\{u \in \mathcal S_{d,n}:\,\frac{\|v-u\|}{(\log n)^{1/d}} \in I\Big\}.
\]

\begin{lemma}[Past witness block]
\label{lem:past_witness_block}
Let $v\in B$ and let $\mathcal P=(B_1,\ldots,B_i=B)$ be a proximity chain such that
\[
\sup_{x\in \bigcup_{j<i}B_j}\|x-v\|
\ge r(\log n)^{1/d}.
\]
Let $I=[\alpha,\beta]\subseteq[0,r]$ satisfy
$\beta-\alpha\ge \ell$. Then there exists $j<i$ such that $B_j\subseteq \bar{\mathcal A}_I(v)$. 
\end{lemma}

\begin{proof}
If $\alpha=0$, take $j=i-1$. Since $v\in B_i$ and $\mathcal P$ is a proximity chain,
for every $x\in B_{i-1}$,
\[
\|v-x\|
\le \ell(\log n)^{1/d}
\le \beta(\log n)^{1/d},
\]
and hence $B_{i-1}\subseteq \bar{\mathcal A}_I(v)$.

Assume now that $\alpha>0$, and let
\[
j^*
=
\max\left\{
j<i:
\inf_{x\in B_j}\|v-x\|
\ge\alpha(\log n)^{1/d}
\right\}.
\]
The set is nonempty. Indeed, by assumption there exist $j_0<i$ and
$y\in B_{j_0}$ such that
\[
\|v-y\|\ge r(\log n)^{1/d}.
\]
For every $x\in B_{j_0}$, since consecutive blocks are within $\ell(\log n)^{1/d}$, and $r-\ell \ge \alpha$,
\[
\|v-x\|
\ge \|v-y\|-\|x-y\|
\ge (r-\ell)(\log n)^{1/d}
\ge \alpha(\log n)^{1/d}.
\]

If $j^*<i-1$, by maximality there exists $z\in B_{j^*+1}$ such that $
\|v-z\|<\alpha(\log n)^{1/d}$.
If $j^*=i-1$, take $z=v\in B_i$, for which the same inequality holds.

For every $x\in B_{j^*}$, since $\mathcal P$ is a proximity chain,
\[
\|v-x\|
\le
\|v-z\|+\|z-x\|
\le
(\alpha+\ell)(\log n)^{1/d}
\le
\beta(\log n)^{1/d}.
\]
By the definition of $j^*$, $\|v-x\|\ge\alpha(\log n)^{1/d}$.
Hence $B_{j^*}\subseteq \bar{\mathcal A}_I(v)$. Since $j^*<i$, this block
belongs to the past of $v$.
\end{proof}

\begin{lemma}[Spread proximity chain] \label{lem:spread_chain} Assume that Algorithm \ref{alg:exact-recovery} is used. Then with high probability, at every step in Phase II of the algorithm, there exists a not fully labeled block $B$ and fully labeled, $\delta$-spread blocks $B_1,\dots,B_i$  such that $(B_1,\dots,B_i,B)$ is a proximity chain and for all $v \in B$,
\begin{align}
\sup_{x\in \bigcup_{j\le i}B_j}\|x-v\|
\ge r(\log n)^{1/d}.\label{eq:length}
\end{align}
\end{lemma}
\begin{proof}
Note that the desired property holds if all blocks in $B_{\mathrm{init}}$ are $\delta$-spread labeled, the graph $H_{\mathrm{spread}}$ is connected and every block is $\ell$-close to a $\delta$-spread block. Let $\mathcal C$ be the connected component of fully labeled
$\delta$-spread blocks containing the initial spread blocks. If Phase~II
has not terminated, choose a not fully labeled block $B$ that is
$\ell$-close to $\mathcal C$. Since $H_{\mathrm{spread}}$ is connected,
there is a proximity chain from the initial spread core to $B$ whose
intermediate spread blocks belong to $\mathcal C$. Choosing the initial
block sufficiently far from $B$ yields \eqref{eq:length}. This event occurs with high probability by Proposition \ref{prop:connectivity} and Lemma~\ref{lemma:poisson_concentration}.
\end{proof}

We introduce a partial ordering on $V$ and write $u<v$ for vertices $u,v \in V$ if $u$ is labeled before $v$ by the estimator $\widehat \sigma$. Let $\widehat V_v:=\{u \in V:\,u<v\}$ be the set of nodes labeled before $v$.

For a vertex $v$ labeled during Phase~II, let $S(v)$ denote the
source actually used to label $v$. Thus, for
Algorithm~\ref{alg:exact-recovery},
\[
S(v):=\widehat V_v\cap N(v),
\]
whereas for Algorithm~\ref{alg:phase2_block},
\[
S(v):=
N(v)\cap
\bigcup_{j\in \widehat V_{\mathrm{cert}}^\dagger(v)}
(V\cap B_j),
\]
where $\widehat V_{\mathrm{cert}}^\dagger(v)$ denotes the collection
of certified blocks available immediately before the block containing
$v$ is processed.

\begin{lemma}[Enough affinity information]
\label{lem:enough_info}
There exists $c_2>0$ such that the following holds.
Fix $m\in\mathbb N$ and assume that
Algorithm~\ref{alg:exact-recovery} or
Algorithm~\ref{alg:phase2_block} is used.
For $v\in V$, let
\[
\mathcal G_v
:=
\bigcap_{w\in V}
\left\{
\left|
\left\{
u\in\widehat V_v\cap N(w):
\widehat\sigma(u)
\notin
\{\omega^\star\circ\sigma^\star(u),*\}
\right\}
\right|
\le m
\right\}.
\]
Then
\[
\P\left(
\bigcup_{v\in V\setminus V_{\mathrm{init}}}
\left\{
\mathcal G_v
\cap
\left\{
\mathcal I_{S(v)}(v)\le c_2\log n
\right\}
\right\}
\right)
=o(1).
\]
\end{lemma}

\begin{proof}
Work on the high-probability event of
Lemma~\ref{lem:spread_chain}. Suppose that a vertex $v$ is labeled
during Phase~II and that $\mathcal G_v$ and $\mathcal I_{S(v)}(v)\le c_2\log n$
hold. We distinguish the two algorithms.

\medskip
\noindent
\emph{Algorithm~\ref{alg:exact-recovery}.}
Let $B$ and fully labeled $\delta$-spread blocks
$B_1,\ldots,B_i$ be supplied by
Lemma~\ref{lem:spread_chain}, and choose any unlabeled vertex
$u\in B$.
Since Algorithm~\ref{alg:exact-recovery} chooses $v$ to maximize
the affinity information among all currently unlabeled vertices,
\[
\mathcal I_{\widehat V_v\cap N(u)}(u)
\le
\mathcal I_{\widehat V_v\cap N(v)}(v)
=
\mathcal I_{S(v)}(v)
\le
c_2\log n.
\]
Define
\[
T(u)
:=
\widehat V_v\cap N(u)
\cap
\bigcup_{j=1}^i B_j.
\]
Since $T(u)\subseteq \widehat V_v\cap N(u)$, monotonicity of the affinity information in its source gives
\[
\mathcal I_{T(u)}(u)
\le
\mathcal I_{\widehat V_v\cap N(u)}(u)
\le
c_2\log n.
\]

\medskip
\noindent
\emph{Algorithm~\ref{alg:phase2_block}.}
Let $B$ be the block containing $v$. By construction of
Algorithm~\ref{alg:phase2_block}, $B$ is $\ell$-close to an already
labeled certified block. Iterating the predecessor relation gives
fully labeled certified blocks $B_1,\ldots,B_i$ such that
$(B_1,\ldots,B_i,B)$ is a proximity chain.

Extending this chain inside the initially labeled certified core if
necessary, we may assume that
\[
\sup_{x\in\bigcup_{j=1}^i B_j}\|x-v\|
\ge r(\log n)^{1/d}.
\]

Then, on $\mathcal G_v$, every fully labeled $\delta$-spread block
$B_j$ contains at most $m$ incorrectly labeled vertices. Hence, for
$n$ sufficiently large,
\[
\sum_{w\in W_{ab}(I_{ab})}
\left|
\left\{
x\in V\cap B_j:
\widehat\sigma(x)=w
\right\}
\right|
\geq
\frac34\delta\pi_{\mathrm{wit}}\log n-m
\geq
\frac12\delta\pi_{\mathrm{wit}}\log n
\]
for every $a\neq b$.

By Lemma~\ref{lem:past_witness_block}, for every $a\neq b$ there
exists $j\leq i$ such that
\[
B_j\subseteq\bar{\mathcal A}_{I_{ab}}(u).
\]

Fix $t=1/2$ and $\Phi\in(0,1)$. For every $a\neq b$, let $B_{j(a,b)}$ be the block supplied by
Lemma~\ref{lem:past_witness_block}. Then
\[
\mathcal I_{T(u)}(u)
\ge
\min_{a\neq b}
\sum_{x\in V\cap B_{j(a,b)}}
-\log
\phi_{1/2}\left(
\bar P_{a\widehat\sigma(x)}(\|u-x\|),
\bar P_{b\widehat\sigma(x)}(\|u-x\|)
\right).
\]
Then, if $\mathcal I_{T(u)}(u)
\leq c_2\log n$, then, for some $a\neq b$, at most
\[
\frac{c_2}{\log(1/\Phi)}\log n
\]
of the correctly labeled $W_{ab}(I_{ab})$-vertices in $B_j$ can
satisfy
\[
\phi_{1/2}\left(
\bar P_{a\hat\sigma(x)}(\|u-x\|),
\bar P_{b\hat\sigma(x)}(\|u-x\|)
\right)
\leq\Phi.
\]
Consequently, at least
\[
\left(
\frac12\delta\pi_{\mathrm{wit}}
-
\frac{c_2}{\log(1/\Phi)}
\right)\log n
\]
vertices in $B_{j(a,b)}$ satisfy
\[
\phi_{1/2}\left(
\bar P_{a\hat\sigma(x)}(\|u-x\|),
\bar P_{b\hat\sigma(x)}(\|u-x\|)
\right)
>\Phi.
\]
Thus, with 
\[
S(u,\Phi)
:=
\bigcup_{a\neq b}
\bigcup_{w\in W_{ab}(I_{ab})}
\left\{
x\in\bar{\mathcal A}_{I_{ab}}(u):
\phi_{1/2}\left(
\bar P_{aw}(\|u-x\|),
\bar P_{bw}(\|u-x\|)
\right)>\Phi
\right\},
\]
we have that
\[
V(S(u,\Phi)) \ge \left(
\frac12\delta\pi_{\mathrm{wit}}
-
\frac{c_2}{\log(1/\Phi)}
\right)\log n,
\]
which is at least $\delta \pi_{\mathrm{wit}} \log n/4$ for $c_2<\delta \pi_{\mathrm{wit}}\log (1/\Phi)/4$.
Next, by Lemma \ref{lem: informative interval}, for any $\epsilon$ there exists some $\Phi<1$ such that, uniformly in $u$, $|S(u,\Phi)|\le \epsilon \log n$.

Thus, by the Mecke equation (Proposition \ref{thm:mecke}), 
\begin{align}
&\P\Big(\bigcup_{u \in V\setminus V_{\mathrm{init}}} \Big\{V(S(u,\Phi)) \ge \frac 14 \delta \pi_{\mathrm{wit}} \log n \Big\}\Big)\nonumber\\
&\le \E \Big[\sum_{u \in V\setminus V_{\mathrm{init}}} \1\Big\{V(S(u,\Phi)) \ge 
\frac14\delta\pi_{\mathrm{wit}}
\log n\Big\}\Big]\nonumber\\
&= \lambda\int_{\mathcal S_{d,n}\setminus B_{\mathrm{init}}} \P\Big(V(S(u,\Phi)) \ge 
\frac14\delta\pi_{\mathrm{wit}}
\log n\Big)\,\mathrm d u.\label{eq:SuPhi}
\end{align}

Further note that $V(S(u,\Phi))$ is Poisson($\lambda |S(u,\Phi)|$)-distributed. Hence, \eqref{eq:SuPhi} is bounded by
\[
\lambda |\mathcal S_{d,n}\setminus B_{\mathrm{init}}| \,\P(\mathrm{Poisson}(\lambda \epsilon \log n) \ge \frac 14 \delta \pi_{\mathrm{wit}} \log n\Big),
\]
which is of order $o(1)$ for $\epsilon>0$ small enough.
\end{proof}

\begin{proof}[Proof of Proposition \ref{prop:abs_propagation}]
Let $c_2>0$ be the constant from Lemma~\ref{lem:enough_info}. Let
\[
\mathcal G_v:= \bigcap_{w\in V}
\{
\left|
\left\{
u\in V_v\cap N(w):
\widehat\sigma(u)
\notin
\{\omega^\star\circ\sigma^\star(u),*\}
\right\}
\right|
\le m
\}.
\]
We
define the auxiliary events
\begin{align*}
\mathcal E_0
&\coloneqq
\bigcap_{u\in V_{\mathrm{init}}}
\left\{
\widehat\sigma(u)=\omega^\star\circ\sigma^\star(u)
\right\},\\
\mathcal H &\coloneqq \bigcap_{v\in V \setminus V_{\mathrm{init}}}
\left(
\mathcal G_v^c
\cup
\left\{
\mathcal I_{S(v)}(v)\ge c_2\log n
\right\}
\right),\\
\mathcal A &\coloneqq \{\text{all blocks contain at most }\Delta \log n\text{ vertices} \}.
\end{align*}
Then $\P(\mathcal E_0 \cap \mathcal H \cap \mathcal A)=1-o(1)$. Given $V$, let \(
\mathcal F_v
\coloneqq
\sigma\left(
V,\sigma^\star,
\{\widehat\sigma(u):u<v\}
\right)\) for $v \in V$. Then $\mathcal G_v$ and $\{
\mathcal I_{S(v)}(v)\ge c_2\log n\}$ are $\mathcal F_v$-measurable. Lemma~\ref{lemma:one_step_error_bound} implies that, on $\mathcal G_v\cap \{ \mathcal I_{S(v)}(v)\ge c_2\log n \}$,
we have
\begin{equation}
\label{eq:conditional-error}
\P\left(
\widehat\sigma(v)
\notin
\{\omega^\star\circ\sigma^\star(v),*\}
\,\middle|\,
\mathcal F_v
\right)
\le
(k-1)\exp\left\{
2\rho m-\mathcal I_{S(v)}(v)
\right\}
\le
(k-1)e^{2\rho m}n^{-c_2}=:p_n.
\end{equation}
Now fix $w\in V$, let $\xi_v
\coloneqq
\1_{\mathcal G_v}
\1\left\{
\mathcal I_{S(v)}(v)\ge c_2\log n
\right\}
\1\left\{
\widehat\sigma(v)
\notin
\{\omega^\star\circ\sigma^\star(v),*\}
\right\}$ and let
\[
I=\{v_1,\ldots,v_{m+1}\}
\subseteq V\cap N(w),
\qquad
v_1<\cdots<v_{m+1}.
\]
Note that \(\prod_{i=1}^{m}\xi_{v_i}
\in
\mathcal F_{v_{m+1}}\). Using the tower property and \eqref{eq:conditional-error},
\[
\E\left[
\prod_{i=1}^{m+1}\xi_{v_i}
\,\middle|\,V
\right]
=
\E\left[
\prod_{i=1}^{m}\xi_{v_i}
\E\left[
\xi_{v_{m+1}}
\mid
\mathcal F_{v_{m+1}}
\right]
\,\middle|\,V
\right]
\le
p_n
\E\left[
\prod_{i=1}^{m}\xi_{v_i}
\,\middle|\,V
\right].
\]
Iterating this argument yields
\[
\label{eq:product-bound}
\E\left[
\prod_{v\in I}\xi_v
\,\middle|\,V
\right]
\le
p_n^{m+1}.
\]
It follows that
\[
\P\left(
\sum_{v\in V\cap N(w)}\xi_v>m
\,\middle|\,V
\right)
\le
\sum_{\substack{
I\subseteq V\cap N(w)\\
|I|=m+1
}}
\E\left[
\prod_{v\in I}\xi_v
\,\middle|\,V
\right]
\le
\binom{|V\cap N(w)|}{m+1}p_n^{m+1}.
\]
On $\mathcal A$ we have that for some constant $C>0$,
\[
\sup_{w\in\mathcal S_{d,n}}
|V\cap N(w)|
\le C\log n.
\]
Hence, uniformly in $w$,
\[
\binom{V(N(w))}{m+1}p_n^{m+1}
\le
(C\log n)^{m+1} (k-1)^{m+1}
e^{2\rho m(m+1)}
n^{-c_2(m+1)}.
\]
We choose the constant $m$ sufficiently large that $c_2(m+1)>1$. In fact, choosing $m$ so that $c_2(m+1)>1+\varepsilon$ for some
$\varepsilon>0$ gives
\begin{equation}
\label{eq:local-tilde}
\sup_{w\in\mathcal S_{d,n}}
\P\left(
\sum_{v\in V\cap N(w)}\xi_v>m
\right)
=o(n^{-1}).
\end{equation}

It remains to show that the auxiliary variables $\xi_v$ control the
total number of propagation errors in every neighborhood.

Suppose that the conclusion of Proposition~\ref{prop:abs_propagation}
fails. Let
\[
v_1<v_2<\cdots
\]
denote the order in which vertices are labeled during Phase~II, and define
\[
t^\star
:=
\min\left\{
t:
\exists\, w\in V
\text{ such that }
\left|
\left\{
v_j:
j\le t,\,
v_j\in N(w),\,
\widehat\sigma(v_j)
\notin
\{\omega^\star\circ\sigma^\star(v_j),*\}
\right\}
\right|
\ge m+1
\right\}.
\]
Set
\[
v^\star:=v_{t^\star}.
\]

By the minimality of $t^\star$, before each vertex $v_j$ with
$j\le t^\star$ is labeled, every neighborhood contains at most $m$
incorrectly labeled vertices. Hence
\[
\mathcal G_{v_j}
\qquad\text{holds for every }j\le t^\star.
\]

Let $w\in V$ be a vertex witnessing the definition of $t^\star$, and
choose distinct vertices
\[
u_1,\ldots,u_{m+1}\in
\{v_1,\ldots,v_{t^\star}\}\cap N(w)
\]
such that
\[
\widehat\sigma(u_j)
\notin
\{\omega^\star\circ\sigma^\star(u_j),*\},
\qquad
j=1,\ldots,m+1.
\]
Since $\mathcal G_{u_j}$ holds for every $j$, on the event $\mathcal H$
we also have
\[
\mathcal I_{S(u_j)}(u_j)\ge c_2\log n.
\]
Therefore
\[
\xi_{u_j}=1,
\qquad
j=1,\ldots,m+1,
\]
and consequently
\[
\sum_{u\in V\cap N(w)}\xi_u>m.
\]

It follows that
\[
\mathcal E_0\cap\mathcal H\cap\mathcal A
\cap
\left\{
\exists\,w\in V:
\left|
\left\{
u\in V\cap N(w):
\widehat\sigma(u)
\notin
\{\omega^\star\circ\sigma^\star(u),*\}
\right\}
\right|
>m
\right\}
\]
is contained in
\[
\bigcup_{w\in V}
\left\{
\sum_{u\in V\cap N(w)}\xi_u>m
\right\}.
\]
Hence
\begin{align*}
&\mathbb P\left(
\exists\,w\in V:
\left|
\left\{
u\in V\cap N(w):
\widehat\sigma(u)
\notin
\{\omega^\star\circ\sigma^\star(u),*\}
\right\}
\right|
>m
\right)
\\
&\qquad\le
\mathbb P\bigl(
(\mathcal E_0\cap\mathcal H\cap\mathcal A)^c
\bigr)
+
\mathbb P\left(
\bigcup_{w\in V}
\left\{
\sum_{u\in V\cap N(w)}\xi_u>m
\right\}
\right).
\end{align*}
The first term is $o(1)$. Moreover, since $|V|\le \Delta n$ on
$\mathcal A$, \eqref{eq:local-tilde} and a union bound give
\[
\mathbb P\left(\mathcal A \cap
\bigcup_{w\in V}
\left\{
\sum_{u\in V\cap N(w)}\xi_u>m
\right\}
\right)
=o(1).
\]
Therefore
\[
\mathbb P\left(
\bigcap_{w\in V}
\left\{
\left|
\left\{
u\in V\cap N(w):
\widehat\sigma(u)
\notin
\{\omega^\star\circ\sigma^\star(u),*\}
\right\}
\right|
\le m
\right\}
\right)
=
1-o(1),
\]
which proves the proposition.
\end{proof}

Finally, we conclude that after Phase II, there are at most $\beta \log n$ incorrectly labeled vertices in the neighborhood of any vertex. Note that, in contrast to the previous statements, the following corollary also allows us to control the number of vertices labeled $*$. 
\begin{corollary}
\label{cor:abs_almost_exact}
Assume that Algorithm \ref{alg:exact-recovery} or \ref{alg:phase2_block} is used. 
    For all $\beta>0$ we have that
    $$
    \P\Big(\bigcap_{u \in V} \{|\{v \in V \cap N(u):\, \hat \sigma(v) \neq \omega^* \circ \sigma^*(v)\}|\le \beta \log n\}\Big)=1-o(1).
    $$
\end{corollary}

\begin{proof}
    The neighborhood of every $u \in \mathcal S_{d,n}$ intersects at most $K_r$ blocks, either $\delta$-occupied or not. With high probability, after Phase II the label of every vertex in a $\delta$-occupied block is in $Z$ (i.e. not $*$). For Algorithm \ref{alg:exact-recovery}, this is trivial since every node $v \in V$ obtains a label in $Z$. For Algorithm \ref{alg:phase2_block}, this follows from the fact that on $\bigcap_{v \in V} \mathcal G_v$, every block in $H_{\mathrm{spread}}$ is fully labeled. Hence, the number of vertices labeled $*$ is bounded by $K_r \delta \log n$. Hence, for all $u \in  V$, with high probability, and for sufficiently large $n$ since $K_r\delta \leq \beta/2$,
\begin{equation}
    \left| \left\{ v\in V \cap N(u): \widehat\sigma(v) \neq \omega^\star\circ\sigma^\star(v) \right\} \right|  \leq K_r\delta \log n + m \leq \beta \log n.
\end{equation}

We now prove almost exact recovery. Let $ M \coloneqq |\left\{
        v\in V:
        \widehat\sigma(v)
        \neq
        \omega^\star\circ\sigma^\star(v)
    \right\}|$.
With high probability, since  for sufficiently large $n$ we have $m<\frac{\beta}{2} \log n$,
\begin{equation}
     M\leq \left\lceil\frac{n}{\ell^d\chi\log n}\right\rceil \delta\log n\leq
    \frac{\delta}{\ell^d\chi}n+\delta\log n.
\end{equation}
By the Poisson Chernoff bound (Lemma \ref{lemma:poisson_concentration}),  $\mathbb P\left(| V|\geq\lambda n/{2}\right)=1-o(1)$.
Then, since $\omega^\star$ is permissible, with high probability and for sufficiently large $n$,
\begin{equation}
    1-A(\widehat\sigma,\sigma^\star)
    \leq \frac{M}{|V|}
     \leq
    \frac{2\delta}{\lambda \ell^d\chi}
    +\frac{2\delta\log n}{\lambda n}.
\end{equation}
Since $\delta$ can be arbitrarily small, this shows almost-exact recovery.
\end{proof}

\section{Refinement}

The proof of exact recovery in the Refinement step is analogous to \cite[Section C.3]{gaudio2026jan}. For completeness, we reproduce the argument. 

\begin{lemma}[Lemma 9 in \cite{gaudio2026jan}]
\label{lem:correct_labels_MAP_failures}
    If $\lambda \nu_d r^d \min_{a \neq b} D_+(P_a \Vert P_b) > 1$, then for a fixed 
    \[0 < \epsilon \leq \frac 12(\lambda \nu_d r^d \min_{a \neq b} D_+(P_a \Vert P_b) - 1 ),\]
    we have that
    \begin{equation*}
        \P(\ell_i(v, \sigma^*) - \ell_j(v, \sigma^*) \leq \epsilon \log n \mid \sigma^*(v) = i) = n^{-(1 + \Omega(1))}
    \end{equation*}
    where
    \begin{equation*}
        \ell_i(v, \sigma) = \sum_{\substack{u \in N(v)}\\{\hat \sigma(u) \neq *}} \log(\bar{p}_{i, \sigma(u)}(x_{uv};\|u - v\|)).
    \end{equation*}
\end{lemma}

We finally complete the proof of Theorem \ref{theorem: achiev} following the proof of Theorem 2 in \cite{gaudio2026jan} (see also the proof of Theorem E.3 in \cite{gaudio2025incguan}).
\begin{proof}[Proof of Theorem \ref{theorem: achiev}]
Recall that $\omega^{\star}$ is such that $\hat{\sigma}_0 = \omega^{\star} \circ \sigma^{\star}$, where $\hat{\sigma}_0$ is the initial block labeling.
Fix $\beta > 0$ and let the event $\mathcal E_1$ be the event that $\hat{\sigma}$ makes at most $\beta \log n$ mistakes in the neighborhood of every vertex $v \in V$. That is,
\begin{equation*}
    \mathcal{E}_1 = \bigcap_{v \in V}\Big\{ |\{u \in V \cap N(v): \hat{\sigma}(u) \neq \omega^* \circ \sigma^*(u) \}| \leq \beta \log n \Big\}.
\end{equation*}
By Corollary \ref{cor:abs_almost_exact}, $\P(\mathcal E_1) = 1 - o(1)$.

Rather than analyze $\hat{\sigma}$ directly, we will provide a uniform upper bound on the failure probability of the MAP estimator on a given vertex $v \in V$, over all labelings $\sigma$ which are ``close'' to the true labels $\sigma^*$. Thus, fix $v \in V$ and let $W(v, \beta)$ be the set of labelings $\sigma$ which differ from $\omega^* \circ \sigma^*$ by at most $\beta \log n$ vertices in the neighborhood of $v$, meaning that
\begin{equation*}
    W(v; \beta) := \{\sigma: |\{u \in V \cap N(v): \sigma(u) \neq \omega^* \circ \sigma^*(u)\}| \leq \beta \log n\}.
\end{equation*}

Now, we define the event
\begin{equation*}
    \mathcal E_v = \bigcup_{i \in Z} \bigg[ \{\omega^* \circ \sigma^*(v) = i\} \cap \bigg( \bigcup_{\sigma \in W(v; \beta)} \bigcup_{j \neq i} \{\ell_i(v, \sigma) \leq \ell_j(v, \sigma) \}\bigg)\bigg],
\end{equation*}
which is the event that there exists a labeling $\sigma \in W(v; \beta)$ such that the MAP estimator with respect to $\sigma$ fails on $v$. 
Note that the probabilty of exact recovery is bounded by
\[
\P(A(\tilde \sigma,\sigma^*)=1) \ge 1-\P\Big(\bigcup_{v \in V} \mathcal E_v\Big).
\]
By the Mecke equation (Proposition \ref{thm:mecke}),
\[
\P\Big(\bigcup_{v \in V} \mathcal E_v\Big)\le \lambda \int_{S_{d,n}} \P(\mathcal E_v) \,\mathrm d v.
 \]

In order to bound $\mathbb{P}(\mathcal{E}_v)$, we leverage Assumption \ref{ass:bounded-ll} to obtain that 
\begin{equation*}
    |\ell_q(v, \omega^* \circ \sigma^*) - \ell_q(v, \sigma)| \leq \rho \beta \log n,\quad q=i,j.
\end{equation*}
for all $\sigma \in W(v; \beta)$. Thus, the event
$\bigcup_{\sigma \in W(v; \beta)}  \{\ell_i(v, \sigma) \leq \ell_j(v, \sigma) \}$ occurring within $\mathcal{E}_v$ implies 
\begin{equation}
\{\ell_i(v,\sigma^{\star}) \leq \ell_j(v, \sigma^{\star}) + 2 \rho \beta \log n\}. \label{eq:refine-implication}
\end{equation} 
Note that $\beta > 0$ can be arbitrarily small.
Conditioned on $\sigma^{\star}(v) = i$, Lemma \ref{lem:correct_labels_MAP_failures} implies that \eqref{eq:refine-implication} occurs with probability $n^{-(1 + \Omega(1))}$ when $\beta > 0$ is small enough, thus establishing that $\mathcal{E}_v$ holds with probability $n^{-(1 + \Omega(1))}$ also. Since $|\mathcal S_{d,n}| \in \Theta(n)$, this finishes the proof.
\end{proof}
\end{document}